\let\accentvec\vec               % llncs with amsmath bugfix
\documentclass{llncs}
\let\vec\accentvec
\usepackage{amsmath}
\usepackage{stmaryrd}
\usepackage{amssymb}
\usepackage{graphicx}
\usepackage{MnSymbol}
\usepackage{url}
\usepackage{mathpartir}
\usepackage{cleveref}
\usepackage{cite}
\usepackage{lscape}
\usepackage[applemac]{inputenc}

\newcommand{\WeakUntil}[2]{#1 \mathrel{W} #2}

\newcommand{\specCorrect}{\mathit{correct}}

\newcommand{\specEVC}{\mathit{correctEVC}} 
\newcommand{\specREG}{\mathit{correctREG}} 

\newcommand{\DThreads}{D_{\textsc{Threads}}}
\newcommand{\DTObs}{D_{\textsc{Threads} \uplus \{ \mathit{obs} \}}}

\newcommand{\seqCons}{\mathit{seqCons}} 
\newcommand{\Lin}{\mathit{Lin}}

\newcommand{\tab}{\hspace*{1.5em}}

\usepackage{ifthen}

\begin{document}
\pagestyle{plain}

\title{An Epistemic Perspective on Consistency of Concurrent Computations}

\author{Klaus v. Gleissenthall\inst{1} \and Andrey Rybalchenko\inst{1,2}}
\institute{Technische Universität M\"unchen \and Microsoft Research Cambridge}
%           {gleissen@in.tum.de rybal@in.tum.de}
\maketitle
\begin{abstract}

Consistency properties of concurrent computations, e.g., sequential
consistency, linearizability, or eventual consistency, are
essential for devising correct concurrent algorithms.  In this paper,
we present a logical formalization of such consistency properties that
is based on a standard logic of knowledge.
Our formalization provides a declarative perspective on what is
imposed by consistency requirements and provides some interesting
unifying insight on differently looking properties.

% Reasoning about consistency properties like sequential consistency,
% linearizability, TSO consistency, or eventual consistency is usually
% perfomed on a per-property basis. As a consequence, advances in
% dealing with one property usually do not benefit advancements in 
% others. In this paper, as a first step towards a unified treatment of
% consistency conditions, we present a logical formalization that is
% based on a standard logic of knowledge. Our formalization provides a 
% declarative perspective is imposed by on what is imposed by consistency requirements
% and provides some inte
%resting unifying insights on apparently
% differently looking properties.

% Consistency properties of concurrent computation, e.g., sequential
% and TSO consistency, linearizability, or eventual consistency, are
% usually formulated operationally in terms of permutations of traces.
% In this paper, we present a logical formalization of such
% consistency properties that is based on a standard logic of
% knowledge.  Our formalization provides a declarative perspective on
% what is imposed by consistency requirements and provides some
% interesting unifying insights on apparently differently looking
% properties. 
\end{abstract}

\section{Introduction}
\label{sec-intro}

% \iffalse
% \newcommand{\specSC}{\mathit{SEQ}} 
% \newcommand{\specSC}{\mathit{SEQ}} 
% \newcommand{\specLIN}{\mathit{LIN}} 
% \newcommand{\specTSO}{\mathit{TSO}} 
% \newcommand{\specEVC}{\mathit{EVC}} 
% \fi

Writing correct distributed algorithms is notoriously difficult. 
While in the sequential case, various techniques for proving
algorithms correct exist \cite{Hoare69,OHearn01}, in the
concurrent setting, due to the nondeterminism induced by
scheduling decisions and transmission failures,
it is not even obvious what correctness actually means.
Over the years, a variety of different \emph{consistency properties} 
restricting the amount of tolerated nondeterminism have been
proposed~\cite{Herlihy08,Herlihy87,Herlihy90,Lamport79,Papadimitriou79,shapiro09}. These
properties range from simple properties like sequential
consistency~\cite{Lamport79} or
linearizability~\cite{Herlihy87,Herlihy90} to complex conditions like
eventual consistency~\cite{shapiro09}, a distributed systems
condition.
Reasoning about these properties is a difficult, yet important task
since their implications are often surprising.% -- even to experienced
%programmers.

Currently, the study of consistency properties and the
development of reasoning tools and techniques for such
properties~\cite{Burckhardt10,Vafeiadis10,Qadeer09} is done for each property individually, i.e., on a
per property basis. To some extent, this trend might be traced back to
the way consistency properties are formulated. 
Typically, they explicitly require existence of certain computation
traces that are obtained by rearrangement of the trace that is to be
checked for consistency,
i.e., these descriptions of consistency properties do not rely on a
logical formalism. While such an approach provides fruitful grounds for the design of
specialized algorithms and efficient tools, it leaves open important
questions such as how various properties relate to each other or
whether advances in dealing with one property can be leveraged for
dealing with other properties.

In contrast to the trace based definitions found in literature, we propose
to study consistency conditions in terms of epistemic
logic~\cite{Fagin03,Halpern90}. 
Here we can rely on a distributed knowledge modality~\cite{Fagin03}, 
which is a natural fit for describing distributed computation.
In this logic, an application $D_G (\varphi)$ of the distributed
knowledge modality to a formula $\varphi$ denotes the fact that a
group $G$ \emph{knows} that a formula $\varphi$ holds.

We present a logical formalization of three consistency properties:
the classical sequential consistency~\cite{Lamport79} and
linearizability~\cite{Herlihy87,Herlihy90}, as well as a recently
proposed formulation~\cite{evCons} of eventual
consistency for distributed databases~\cite{shapiro09}.  
Our characterizations show that moving the viewpoint from reasoning
about traces (models) to reasoning about knowledge (logic) can lead to
new insights.
When formulated in the logic of knowledge, these differently looking
properties agree on a common schematic form: $\neg D_G (\neg
\specCorrect)$.  
According to this schematic form, a computation satisfies a
consistency property if and only if a group $G$ of its participants,
i.e., threads or distributed nodes, \emph{do not know that} the
computation violates a specification~$\specCorrect$ that describes
computations from the \emph{sequential} perspective, i.e., without
referring to permutations thereof.
For example, when formalising sequential consistency of a concurrent
register $\specCorrect$ only states that the first read operation
returns zero and any subsequent read operation returns the value
written by the latest write operation.

The common form of our characterizations exposes the differences
between the consistency properties in a formal way.
A key difference lies in the group of participants that provides
knowledge for validating the specification~$\specCorrect$. 
For example, a computation is sequentially consistent if it satisfies
the formula $\neg D_\textsc{Threads}(\neg \specCorrect)$, i.e., the group
$G$ of agents needed to validate the sequential specification
comprises the group of threads $\textsc{Threads}$ accessing the shared
memory. 
Surprisingly, the same group of agents is needed to validate eventual
consistency, since in our logic it is characterized by the formula
%
%\begin{center}
%  \begin{minipage}[t]{.9\columnwidth}
    $\neg D_\textsc{Threads}(\neg \specEVC)$. 
 % \end{minipage}
%\end{center}
%
This reveals an insight that eventual consistency is actually not an
entirely new consistency condition, but rather an instance of
sequential consistency that is determined by a particular choice
of~$\specCorrect$. 
In contrast to the two above properties, the threads' knowledge is not
enough to validate linearizability. 
To capture linearizability, the set of participants $G$ needs to go
beyond the participating threads $\textsc{Threads}$ and include an
additional observer thread~$\mathit{obs}$ as well.
The observer only acquires knowledge of the relative order between
returns and calls. As logical characterization of linearizability, we
obtain $\neg D_{\textsc{Threads}\cup\{\mathit{obs}\}}(\neg
\specCorrect)$.

We show that including the observer induces a different kind of
knowledge, i.e., it weakens the modal system from \emph{S5} to
\emph{S4}~\cite{Ditmarsch08}.
As a consequence, the agents lose certainty about their
decision whether or not a trace is consistent. 
For sequential consistency ($\seqCons$) the agents know whether or not a trace is
sequentially consistent, i.e., the formula $ (\seqCons \leftrightarrow
D_\textsc{Threads}(\seqCons)) \mathrel{\land} (\neg\seqCons
\leftrightarrow D_\mathit{Threads}(\neg\seqCons))$ is valid.
In contrast, for linearizability ($\Lin$) the threads cannot be sure whether a trace they
validate as linearizable is indeed linearizable, i.e., there exists a
trace that satisfies $\Lin \wedge \neg \DTObs (\Lin)$.

The discovery that eventual consistency can be reduced to sequential
consistency is facilitated by a generalization of classical sequential
consistency that follows naturally from taking the epistemic
perspective. 
%
%We formalize $\specEVC$ in temporal logic with first-order
%quantification. %Our formalization of $\specEVC$ provides a
%declarative view on the restrictions imposed by eventual consistency.
Our formalization of $\specCorrect$ for eventual consistency is given
by $\specEVC$ that requires nodes to keep consistent logs, i.e.,
whenever a transaction is received by a distributed node, the
transaction must be inserted into the node's logs in a way that is
consistent with the other nodes' recordings.
We allow $\specEVC$ not only to refer to events that are performed by
the nodes that take part in the computation, but also to auxiliary
events that model the \emph{environment} that interacts with nodes. 
We use the environment to model transmission of updates from one
distributed node to another. Our knowledge characterization then
implicitly quantifiers over the order of occurrence of such events,
which serves as a correctness certificate for a given trace.

\paragraph{Contributions}
In summary, our paper makes the following
contributions. 
%\begin{itemize}
We provide characterisations for sequential consistency (Section
\ref{sec:seqCons}), eventual consistency (Section \ref{sec:evCons})
and linearizability (Section \ref{sec:lin}) which we prove correct
wrt. their standard definitions. Our characterizations reveal a
remarkable similarity between consistency properties that is not
apparent in their standard formulations. 
Through our characterizations, we identify a natural generalization of sequential consistency that
allows us to reduce eventual consistency, a complex property usually
defined by the existence of two partial-ordering relations, to
sequential consistency. In contrast to this reduction, we show that
linearizability requires a different kind of knowledge than
sequentially consistency and prove a theorem (Section
\ref{sec:detect}) illustrating the ramifications of this difference.

\section{Examples}
\label{sec:illu}
In this section, before providing technical details, we give an informal overview of our characterizations.
\subsection{Sequential Consistency}
\paragraph*{Trace-Based Definition}
The most fundamental consistency condition that concurrent computation are intuitively expected to satisfy is sequential
consistency \cite{Lamport79}. Its original definition reads:
\begin{quote}
The result of any execution is the same as if the operations of all
the processors were executed in \emph{some sequential order}, and the
operations of each individual processor appear in the sequence in the
\emph{order specified by its program}.
\end{quote}
Equivalently, this more formal version can be found in the literature (cf., \cite{Attiya94}): For a trace $E$ to be sequentially consistent, it needs to satisfy two conditions: (1) $E$ must be \emph{equivalent} to a witness trace $E'$ %(we write $E \sim_{SC} E'$) 
and (2) trace $E'$ needs to be \emph{correct} with respect to some specification. %(we write $E' \models \specCorrect$).
To be equivalent, two traces need to be permutations that preserve the local order of events for each thread.
\begin{example}
\label{ex:seq}
Consider the following traces representing threads $t_1$ and $t_2$ storing and loading values on a shared register.
For the purpose of this example, we assume the register to be initialized with value~$0$. We use ``$:=$'' to abbreviate ``equals by definitions''.
\begin{equation*}
\begin{array}{rl}
E_1 :=& (t_2,\mathit{ld}(0)) \; (t_2,\mathit{ld}(1)) \; (t_1,\mathit{st}(1))\\[\jot]
E_2 :=& (t_2,\mathit{ld}(0)) \; (t_1,\mathit{st}(1)) \; (t_2,\mathit{ld}(1))\\[\jot]
E_3 :=&(t_2,\mathit{ld}(0)) \; (t_1,\mathit{st}(1)) \; (t_2,\mathit{ld}(2))
\end{array}
\end{equation*} 
Trace $E_1$ is sequentially consistent, because it is
equivalent to $E_2$ and $E_2$ meets the specification of a shared
register, i.e., each load returns the last value stored. In contrast,
$E_3$ is not sequentially consistent, because no appropriate witness
can be found. In no equivalent trace, $t_2$'s load of $2$ is
preceded by an appropriate store operation.
\end{example}
\paragraph{Logic} In this paper, in contrast to the above
trace-based formulation, we investigate consistency from the
perspective of \emph{epistemic logic}. Epistemic logic is a formalism used for reasoning
about the knowledge distributed nodes/threads acquire in a distributed
computation. For example, in trace $E_1$ thread $t_2$
\emph{knows} it first loaded value~$0$ and then value~$1$ while $t_1$
\emph{knows} it stored~$1$.  When we consider the knowledge acquired
by the threads $t_1$ and $t_2$ together as a group, we say that the
group of threads $\{t_1, t_2\}$ \emph{jointly knows} $t_2$ first
loaded~$0$ and then~$1$ while $t_1$ stored~$1$. We denote the
fact that a group $G$ jointly knows that a formula $\varphi$ holds
by~$D_G(\varphi)$, which is an application of the distributed
knowledge modality. 
According to our logical characterization of sequential consistency:
%\begin{equation*}
 $\neg D_\textsc{Threads}(\neg \specCorrect)$,
%\end{equation*}
a trace is sequentially consistent, if the group of all threads accessing the shared data-structure does not jointly know that the trace is not correct. 
\addtocounter{example}{-1}
\begin{example}[continued]
This means trace $E_1$ is sequentially consistent. In trace $E_1$, the threads know that $t_2$ first loaded $0$ and then $1$ and that $t_1$ stored $1$, however they do not know in which order these events were scheduled. This means, for all they know $t_1$ could have stored $1$ before $t_2$ loaded it and after $t_2$ loaded~$1$, which would meet the specification. In contrast, $E_3$ is not sequentially consistent. The threads know that $t_2$ loaded $2$, however no thread stored it. This means $E_3$ cannot have met the specification.
\end{example}

\paragraph{Indistinguishability} We formalize this notion of knowledge in
terms of the local perspective individual threads have on the
computation. We extract this perspective by a function
$\downarrow$ such that $E \downarrow t$ projects trace $E$
onto the local events of thread $t$. 
If two traces do not differ from the local perspective of thread $t$,
we say that they are indistinguishable for $t$. We write $E \sim_t E'$
to denote that for thread $t$, trace $E$ is indistinguishable from
trace $E'$. Combining their abilities to distinguish traces, a group
of threads can distinguish two traces whenever there is a thread in
the group that can. We write $E \sim_G E'$ to denote that for
group~$G$ trace $E$ is indistinguishable from trace $E'$. 
Indistinguishability allows us to define the knowledge of a group. 
A group $G$ knows a fact $\varphi$ if this fact holds on all traces
that the threads in $G$ cannot distinguish from the actual trace. We
write $E \models \varphi$ to say that trace $E$ satisfies formula
$\varphi$. Formally (see Section \ref{def:semantics}):
$E \models D_G (\varphi)$ :iff for all $E'$ s.t. $E \sim_G E'$: $E' \models \varphi$, 
where we use ":iff" to abbreviate "by definition, if and only if".

\addtocounter{example}{-1}
\begin{example}[continued] For trace $E_1$, the thread-local
projections are: $E_1 \downarrow t_1 = (t_1,\mathit{st}(1))$ and $E_1
\downarrow t_2 = (t_2,\mathit{ld}(0)) (t_2,\mathit{ld}(1))$. We
get the same projections for $E_2$, and $E_3 \downarrow t_1 =
(t_1,\mathit{st}(1))$ and $E_3 \downarrow t_2 = (t_2,\mathit{ld}(0))
(t_2,\mathit{ld}(2))$. From these projections, we get: $E_1 \sim_{t_1}$ $E_2 \sim_{t_1} E_3$
and $E_1 \sim_{t_2} E_2$ but $E_1 \not \sim_{t_2} E_3$ and $E_2 \not
\sim_{t_2} E_3$.  For groups of threads, we have: $E_1 \sim_{\{ t_1,t_2 \}} E_2$
but $E_2 \not\sim_{\{ t_1,t_2 \}} E_3$, because $E_2 \not\sim_{t_2}
E_3$. 
We write $E \models \specREG$ to say $E$ is correct with respect
to the specification of a shared register. Then $E_1 \models \neg
\DThreads (\neg \specREG)$, $E_2 \models \neg
\DThreads (\neg \specREG)$ and $E_3 \models \DThreads (\neg
\specREG)$.
\end{example}

\paragraph{Knowledge in the Trace-Based Formulation} Interestingly, the notion of equivalence found in the trace-based formulation of sequential consistency precisely corresponds to $\sim_\textsc{Threads}$, the indistinguishability relation of all threads accessing the shared data-structure. This suggests that the knowledge-based formulation of consistency lies already buried in the original definition. Similarly, the formulation ``The result of any execution is the same as if ...", found in the original definition alludes to the possibility of a fact $\varphi$, which, in epistemic logic, is represented by the dual modality of knowledge $\neg D_G (\neg \varphi)$.

\subsection{Eventual Consistency}
%\paragraph{Eventual Consistency}
Eventual consistency~\cite{shapiro09} is a correctness condition for distributed database systems, as those employed in modern geo-replicated internet services. In such systems, threads (distributed nodes) keep local working-copies (repositories) of the database which they may update by performing a commit operation. Queries and updates have revision ids, representing the current state of the local copy. 
Whenever a thread commits, it broadcasts local changes to its repository and receives changes made by other threads.
After the commit, a new revision id is assigned. As the underlying network is unreliable, committed changes may however be delayed or lost before reaching other threads. 

In this setting, weaker guarantees on consistency than in a multi-processor environment are required, as network partitions are unavoidable, causing updates to be delayed or lost. Consequently, eventual consistency is a prototypical example for what is called ``weak''-consistency. We present a recent, partial-order based definition drawn from the literature~\cite{evCons} in Section \ref{sec:evCons}.

Taking the knowledge perspective on eventual consistency reveals a remarkable insight. Eventual consistency is actually not a new, weaker consistency condition, but just sequential consistency -- with an appropriate sequential specification. 

%\paragraph{Logic}
In our logical characterization, eventual consistency is defined by the formula:
%\begin{center}
$E \models \neg \DThreads (\neg \specEVC$).
%\end{center}
That is, to be eventually consistent, a trace needs to be sequentially consistent with respect to a sequential specification $\specEVC$.  Our formula for $\specEVC$ uses the past time modality $\boxminus (\varphi)$ (see Section \ref{def:semantics}), representing the fact that so far, formula $\varphi$ was true.  
We specify $\specEVC$ by:
\begin{equation*}
\specEVC := 
  \begin{array}[t]{l@{}}
  \begin{array}[t]{l@{}}
  \forall t \forall q \forall r ( \boxminus( \mathit{query}(t,q,r) \rightarrow 
	 \exists \mathcal{L}( \mathcal{L} \; \mathit{validLog} \; t \;  \wedge  \mathit{result}(q,\mathcal{L},r)))) \\ [\jot]
 \wedge \; \mathit{atomicTrans} \wedge \mathit{alive} \wedge \mathit{fwd} 
  \end{array}
\end{array}
\end{equation*}
This formula says that for all threads, queries and results, so far, whenever a thread~$t$ posed a query~$q$ to its local repository, producing result~$r$, thread~$t$ must be able to present a valid log $\mathcal{L}$, such that the result of posing query~$q$ on a machine that performed only the operations logged in log~$\mathcal{L}$ matches the recorded result~$r$. The additional conjuncts  $\mathit{atomicTrans}$, $\mathit{alive}$ and $\mathit{fwd}$ specify further requirements on the way updates may be propagated in the network. 

In our characterization, a log~$\mathcal{L}$ is a sequence of actions (i.e., queries, updates and commits). The formula $\mathit{validLog}$ describes the  conditions a log has to satisfy to be valid for a thread~$t$:
\begin{center}
$\mathcal{L} \; \mathit{validLog}  \; t \; := \; \forall a(a \; \mathit{in} \; \mathcal{L} \leftrightarrow t \; k_{\mathit{log}} \;  a ) \; \wedge \mathit{consistent}(\mathcal{L})$
\end{center}
This formula requires that for all actions~$a$, $a$ is logged in~$\mathcal{L}$ (represented by the infix-predicate $\mathit{in}$) if and only if thread~$t$ knows about action~$a$. A thread knows about all the actions that it performed itself and the actions performed in revisions that were forwarded to it. The formula $\mathit{consistent}({\mathcal{L}})$ ensures that all actions in the log~$\mathcal{L}$ appear in an order consistent with the actual order of events.

\paragraph{Environment Events} 
To make this result possible, we make a generalization that comes natural in the knowledge setting. We allow traces to contain \emph{environment} events that represent actions that are not controlled by the threads that participate in the computation. In our characterization, environment events are used to mark positions where updates were successfully forwarded from one client to another. By allowing $\specEVC$ to refer to those events, we implicitly encode an existential quantification over all possible positions for these events. That means a trace is eventually consistent if any number of such events could have occurred such that the specification $\specEVC$ is met.

\begin{example}
 Consider the following traces of a simple database that allows clients to update and query the integer variable $x$:
\begin{equation*}
\begin{array}{rl}
E_{4} :=&(t_1, \mathit{up}(0,x:=0)) \;  (t_1,\mathit{com}(0))\; (t_1,\mathit{up}(1,x:=1))
 (t_1,\mathit{com}(1))  \\ [\jot] 
& (t_2,\mathit{qu}(0,x,0)) (t_2,\mathit{com}(0)) \; (t_2,\mathit{qu}(1,x,1)) \\  [\jot]
E_{5} :=& (t_1, \mathit{up}(0,x:=0)) \;  (t_1,\mathit{com}(0)) \; (\mathit{env},\mathit{fwd}(t_1,t_2,0))
 (t_1,\mathit{up}(1,x:=1)) \\ [\jot]
 & (t_1,\mathit{com}(1)) \;  (t_2,\mathit{qu}(0,x,0)) 
(t_2,\mathit{com}(0)) \; (\mathit{env},\mathit{fwd}(t_1,t_2,1)) \\ [\jot]
& (t_2,\mathit{qu}(1,x,1))
\end{array}
\end{equation*}
Updates are of the form $\mathit{up}(\mathit{id},u)$, where $\mathit{id}$ is the revision-id and $u$ the actual update. In our example, updates are variable assignments $x:=v$ meaning that a variable~$x$ is assigned value $v$. Queries are of the form $\mathit{qu}(\mathit{id},q,r)$, where $\mathit{id}$ stands for the revision-id, $q$ for the query, and $r$ for the result. Queries in our example consist only of variables, i.e., a query returns the current value assigned. The action $\mathit{com}(\mathit{id})$ represents the act of committing, that is, sending revision $\mathit{id}$ over the network and checking for updates. 
Forwarding actions are performed by the environment $\mathit{env}$. The event $(\mathit{env},\mathit{fwd(t,t',\mathit{id})})$ represents the environment forwarding the changes made in revision $\mathit{id}$ from thread~$t$ to thread~$t'$. 

In trace $E_{5}$, when thread~$t_2$ queries the value of $x$ in revision $0$, thread~$t_2$ can present the log
$\mathcal{L} := \mathit{up}(0,x:=0) \; \mathit{com}(0) \; \mathit{qu}(0,x,0)$
as an evidence of the correctness of the result $0$. As by the time of $t$'s query, only revision $0$ has been forward from $t_1$ to $t_2$, thread $t_2$ only knows about $t_1$'s first update and its own query. Querying $x$ after the update $x:=0$ yields $0$, so $\mathit{result}(x,\mathcal{L},0)$ holds.

When thread~$t_2$ queries $x$ in revision $1$, thread~$t_1$'s second update has been forwarded, so $t_2$ can present the log
$\mathcal{L} :=  \mathit{up}(0,x:=0) \; \mathit{com}(0) \; \mathit{up}(1,x:=1) \;  \mathit{com}(1) \;
\mathit{qu}(0,x,0) \; com(0) \; \mathit{qu}(1,x,1)$.
Since $t_2$ received the $t_1$'s revision~$1$ the log contains the second update $x:=1$ and $t_2$'s query of $x$ returns $1$. This means $E_5 \models \specEVC$.
As a consequence, we have $E_4 \models \neg \DThreads (\neg \specEVC)$, because $E_4 \sim_{\textsc{Threads}} E_5$ and $E_5 \models \specEVC$. The forwarding events in $E_5$ mark positions where the transmission of updates through the network could have occurred to make the computation meet $\specEVC$.
\end{example}

\subsection{Linearizability}
While the threads' knowledge characterizes sequential consistency and eventual consistency, their knowledge is not strong enough to define linearizability. Linearizability extends sequential consistency by the requirement that method calls must effect all visible change of the shared data at some point between their invocation and their return. Such a point is called the \emph{linearization points} of the method. 

To characterize linearizability, we introduce another agent called \emph{the observer} that tracks the available information on linearization points. To do this, the observer monitors the order of non-overlapping (sequential) method calls in a trace. The observer's view of a trace is the order of non-overlapping method calls. This order is represented by a set of pairs of return and invoke events, such that the return took place before the invocation. We extract this order by a projection function $\mathit{obs}(\cdot)$.
\begin{example}
Consider the following traces where method calls are split into invocation- and return events: %Actions are of the form $\mathit{call} \;  m(v)$ and $\mathit{ret} \; m(v)$, representing the threads' calling or returning from a method $m$ with value $v$.
\begin{equation*}
\begin{array}{rl}
E_6:=&(t_2, \mathit{inv} \;  \mathit{ld}()) \; (t_2, \mathit{ret} \; \mathit{ld}(1)) \; (t_1, \mathit{inv} \;  \mathit{st}(1)) \;
(t_1,\mathit{ret} \;  \mathit{st}(\mathit{true}))\\[\jot]
E_7:=&(t_2, \mathit{inv} \; \mathit{ld}()) \; (t_1, \mathit{inv} \;  \mathit{st}(1)) \; (t_2, \mathit{ret} \; \mathit{ld}(1)) \;
(t_1,\mathit{ret} \;  \mathit{st}(\mathit{true})) \\[\jot]
 E_8 :=&(t_1, \mathit{inv} \; \mathit{st}(1)) \; (t_1,\mathit{ret} \;  \mathit{st}(\mathit{true})) \; (t_2, \mathit{inv} \;  \mathit{ld}()) \;  (t_2, \mathit{ret} \; \mathit{ld}(1))
\end{array}
\end{equation*}
For trace $E_6$, the observer's projection function $\mathit{obs}(\cdot)$ yields:
%\begin{equation*}
$\mathit{obs}(E_6) = \{ ( \; (t_2, \mathit{ret} \; \mathit{ld}(1))  ,  (t_1, \mathit{inv} \;  \mathit{st}(1)) \; ) \}$.
%\end{equation*}
This means the observer sees that $t_2$'s load returned before $t_1$'s store was invoked. In trace $E_7$, the method calls overlap. Consequently, the observer knows nothing about this trace:
%\begin{equation*}
$\mathit{obs}(E_7) := \varnothing$.
%\end{equation*}
For $E_8$, we get $\mathit{obs}(E_8) = \{ ( \; (t_1, \mathit{ret} \; \mathit{st}(\mathit{true})) ,  (t_2, \mathit{inv} \;  \mathit{ld}()) \; ) \}$.
\end{example}
The observer's view tracks the available information on linearization points. In trace $E_6$, thread $t_2$'s linearization point for the call to load must have occurred before the linearization point of $t_1$'s call to store. This follows from the fact that $t2$'s load returned before $t1$'s call to store and that linearization point must occur somewhere between a method's invocation and its return. In trace $E_7$ linearization points may have occurred in any order as the method calls overlap. 

To the observer, a trace $E$ is indistinguishable from a trace $E'$ if the order of linearization points in $E$ is preserved in $E'$ and maybe an order between additional linearization points is fixed (see Section \ref{sec:indist}):
%\begin{equation*}
$E \preceq_{obs} E' \; \text{:iff} \; \mathit{obs}({E}) \subseteq \mathit{obs}(E')$.
%\end{equation*}
A trace $E$ is linearizable if the threads together with the observer do not know that the trace is incorrect:
%\begin{equation*}
$E \models \neg \DTObs (\neg \specCorrect)$.
%\end{equation*}

\addtocounter{example}{-1}
\begin{example}[Continued]
We have
%\begin{equation*}
$E_6 \not \preceq_{obs} E_7 \; \text{, but} \; E_7 \preceq_{obs} E_6$.
%\end{equation*}
% Expanding our characterization of linearizability, we get:
% \begin{center}
% $E \models \neg \DTObs \neg \specCorrect$\\
% iff there is $E'$ s.t. $E \sim_{\textsc{Threads} \uplus \{ \mathit{obs}\}} E'$ and $E' \models \specCorrect$.
% \end{center}
% This means a trace $E$ is linearizable if and only if there exists a trace $E'$ such that the threads cannot distinguish $E$ from $E'$ and the order of linearization points in $E$ is preserved in $E'$. 
% In our example trace $E_6$ is linearizable since $E_6 \sim_{\textsc{Threads} \uplus \{ \mathit{obs}\}} E_7$ and $E_7 \models \specREG$. However, trace $E_5$ is not linearizable since we can find no indistinguishable trace that meets the specification.
%We write $E \models \specREG$ to say $E$ is correct with respect to the specification of a shared register. 
Trace $E_7$ is linearizable since $E_7 \sim_{\textsc{Threads} \uplus \{ \mathit{obs}\}} E_8$ and $E_8 \models \specREG$. However, trace $E_6$ is not linearizable since there is no indistinguishable trace that meets the specification. Note that the threads without the observer could not have detected this violation of the specification, i.e., $E_6 \models \neg \DThreads \neg \specREG$.
\end{example}

\subsection{Knowledge about Consistency}
As we describe sequential consistency in a standard logic of
knowledge, corresponding axioms apply (see, e.g. \cite[chapter
2.2]{Ditmarsch08}). 
For example, everything a group of threads knows is also true:
$\text{(T)} := \; \models D_G (\varphi) \rightarrow \varphi \; \text{(Truth axiom)}$,
groups of threads know what they know:
$\text{(4)} := \; \models D_G (\varphi) \rightarrow D_G (D_G (\varphi)) \; \text{ (positive introspection)}$
and groups of threads know what they do not know:
$\text{(5)} := \; \models \neg D_G (\varphi) \rightarrow D_G (\neg D_G (\varphi)) \;  \text{ (negative introspection)}$.
For an complete axiomatization of a similar epistemic logic with time see \cite{Belardinelli08}.

Interestingly, adding the observer not only strengthens the threads' ability to distinguish traces but changes the \emph{kind} of knowledge agents acquire about a computation. Whereas $\sim_\textsc{Threads}$ is an \emph{equivalence relation}, $\sim_{\textsc{Threads} \uplus \{ \mathit{obs}\}}$ is only a \emph{partial order}. As a consequence, $\DThreads$ corresponds to the modal system \emph{S5}, whereas $\DTObs$ corresponds to the weaker system \emph{S4} \cite{Ditmarsch08}. This means, that $\DTObs$ does not satisfy the axiom of negative introspection (5). 

It seems natural to ask if the differences in the type of knowledge between sequential consistency and linearizability affect the ability to detect violations of the specification. In Section \ref{sec:detect}, we show that the difference the lack of axiom (5) makes, lies in the certainty threads have about their decision. Whereas for sequentially consistent ($\seqCons:= \neg \DThreads (\neg correct)$), whenever the threads decide that a trace is sequentially consistent, they can be sure that the trace is indeed sequentially consistent:
%
%\begin{equation*}
%  \begin{array}[t]{@{}l@{}}
    $(\seqCons \leftrightarrow D_\textsc{Threads}(\seqCons))$ %\\[\jot]
%  \end{array}
%\end{equation*}
%
for linearizability ($\mathit{Lin} := \neg D_{\textsc{Threads} \uplus \{\mathit{obs} \}} (\neg \mathit{correct})$), it can occur that the threads together with the observer decide that a trace is linearizable, however, they cannot be sure that it really was:
%
%\begin{equation*}
 $ \Lin \wedge \neg \DTObs (\Lin)$.
%\end{equation*}
%
\section{Logic Of Knowledge}
\label{sec:log}
In this section we present a standard logic of knowledge (see
\cite{Halpern90}) that we use for our characterizations. 
We follow the exposition of \cite{kramer10}. We define the set
$\mathcal{E}$ of events as $\mathcal{E} \ni e := (t,\mathit{act})$, representing $t \in
\textsc{Threads} \uplus \{ \mathit{env}\}$ performing an action $\mathit{act} \in
\mathcal{A}$. The environment $\mathit{env}$ can perfom
synchronization events that go unseen by the threads. In our
characterization of eventual consistency, the environment forwards
transactions from one node to the other.
%By $m_{\mathit{sys}}$ and $d_{\mathit{sys}}$ we denote the memory-, and
%database system, respectively. They perform synchronization operations that go unseen
%by the threads, i.e. flushing values to memory for TSO and forwarding
%transactions from one node to the other for eventual consistency.
%We will instantiate $\mathcal{A}$ for the different
%characterizations. 
We define the generic set of actions:
%\begin{equation*}
$\mathcal{A} \ni \mathit{act}:= \; \mathit{inv}(m,v)  \; | \; \mathit{ret}(m,v)$.
%\end{equation*}
Threads can invoke or return from methods $m \in \textsc{Methods}$ with $v \in
\textsc{Values}$. For our characterization of
eventual consistency, we instantiate $\mathcal{A}$
with application-specific actions. These can easily be translated back into
the generic form by splitting up events into separate invocation-
and return-parts. 
%  We define the generic set of actions:
%  %\begin{equation*}
%  $\mathcal{A} \ni \mathit{act}:= \mathit{call}(m,v) \; | \;
%  \mathit{ret}(m,v)$. 
%  %\end{equation*}
% Threads can call or return from methods $m \in \textsc{Methods}$ with $v \in
% \textsc{Values}$. For our characterizations with fixed specifications
% (i.e., TSO and eventual consistency) we will instantiate $\mathcal{A}$
% with application specific actions. These can however easily be
% translated back into the generic form. 
\subsection{Preliminaries}
\label{sec:prelim}
%For $i \in \mathbb{N}$, we let $[i] = \{1, \ldots, i\}$. %"$\_$" represents irrelevant, existential quantification.
We denote by $\mathcal{E}^\ast$ the set of finite-, and by $\mathcal{E}^\omega$ the set of infinite sequences over $\mathcal{E}$.
We denote the empty sequence by $\epsilon$. Let $\mathcal{E}^\infty :=
\mathcal{E}^\ast \; \uplus \; \mathcal{E}^\omega$ and $E
\in\mathcal{E}^\infty$. Then  $E \downharpoonright i$ denotes the
finite prefix up to- and including $i$. We let $E@i$ be the element of
sequence $E$ at position $i$. We define $\mathit{len}(E)$ to be the
length of $E$, where $\mathit{len}(\epsilon)=0$, and $\mathit{len}(E)=
\omega$, if $E \in E^\omega$. For $e \in \mathcal{E}$, we say that $\mathit{pos}(e,E)= j$, if  $E@j=e$ and $\mathit{pos}(e,E)=\omega$ otherwise. Hence, we write $e \in E$ if $\mathit{pos}(e,E) < \omega$. 

%\subsection{Observations and Indistinguishability}
\label{sec:indist}
We formally define projection functions and indistinguishability
relations. A thread's view of a computation trace is the part of the trace it can
observe. We define this part by a projection function that extracts
the respective events.  We use this projection function to define an
indistinguishability relation for each thread. 

\paragraph*{Thread Indistinguishability Relation}
% A thread $t$ can observe events that represent its calling a method or
% its returning from one. It can distinguish traces by this information
% only. Thus, if two traces share the events concerning thread~$t$,
% thread~$t$ cannot distinguish them. 
For a thread $t \in$ \textsc{Threads} the indistinguishability
relation $ \sim_t \; \subseteq (\mathcal{E}^\omega \times (\mathbb{N}
\uplus \{ \omega )\})^2$ is defined such that: %for all $(E,i),(E',i') \in \mathcal{E}^\omega \times \mathbb{N}$:
$ (E,i) \sim_t (E',i') \mbox{ :iff }  (E \downharpoonright i)
\downarrow t = (E' \downharpoonright i') \downarrow t$
where $\downarrow : (\mathcal{E}^\infty \times \textsc{Threads}) \to
\mathcal{E}^\infty$ designates a projection function onto $t$'s local
perspective. $E \downarrow t$ is the projection on events in the set
$\{ (t,\mathit{act}) \; | \; \mathit{act} \in \mathcal{A} \}$, i.e.,
the sequence obtained from $E$ by erasing all events that are not in
the above set.
% is defined such that: 
% \[ \epsilon \downarrow t = \epsilon  \]
% \[ (e \cdot E) \downarrow t = 
% \begin{cases}
% e \cdot (E \downarrow t), & \text{if } e= (t, \_) \\
% E \downarrow t, & \text{otherwise}
% \end{cases} \]

\paragraph{Observer Indistinguishability Relations}
The observer's view of a trace is the order of non-overlapping method
calls. %This order induces an order on the linearization points of the
%methods. That is, if a method terminates before another method is
%called, the linearization point of the first method must be before
%that of the second. The observer cannot distinguish a trace $E$ from
%any trace $E'$ such that $E'$ respects the order of linearization
%points in $E$ and possibly extends it. 
We let $\textsc{Inv} \ni \mathit{in} :=
(t, \mathit{inv}(m,v))$ and $\textsc{Ret} \ni r :=  (t,
\mathit{ret}(m,v))$. The indistinguishability
relation of the observer $ \preceq_{\mbox{obs}} \; \subseteq
(\mathcal{E}^\omega \times \mathbb{N})^2$ is given by: %t holds that
for all $(E,i),(E',i') \in \mathcal{E}^\omega \times \mathbb{N}$:
 $(E,i) \preceq_{\mbox{obs}} (E',i')\mbox{ :iff }\mbox{obs}(E,i)
\subseteq \mbox{obs}(E' , i')$ where obs: $(\mathcal{E}^\omega \times
\mathbb{N}) \to \mathcal{P}(\mathcal{E}^2)$ designates a projection
onto the observer's local view, such that:
%\[ \mbox{obs}(E,i) =\{ ( E@j, E@k) \; | \;  1 \leq j < k \leq i \;  \wedge \; E@j \in \mathcal{E}_{\mbox{ret}}  \wedge \;
%E@k \in \mathcal{E}_{\mbox{call}} \}  \]
%\begin{sloppypar}
%\begin{center}
 obs$(E,i) =\{ (r,\mathit{in}) \in$ \textsc{Ret} $\times \textsc{ Inv} \; |$  pos$(r,E) < $ pos$(\mathit{in},E) \leq i \}$.
%\end{center}
We abbreviate $\mathit{obs}(E):= \mathit{obs}(E,\mathit{len}(E))$.
%\end{sloppypar}

\paragraph{Joint Indistinguishability Relations}
Joint indistinguishability relations link pairs of traces that a group
of threads can distinguish if they share their knowledge. Whenever a
thread in the group can tell the difference between two traces, the
group can. Let $G \subseteq \textsc{Threads}$. 
We define the joint indistinguishability relation of group $G$ to
be $\sim_{G} := (\bigcap_{t \in G} \sim_t)$ and $\sim_{G \uplus \{
\mathit{obs}\}} := \sim_{G} \cap \preceq_\mathit{obs}$. For any
indistinguishability relation $\sim$, we write $E \sim E'$ as an
abbreviation for $(E,\mathit{len}(E)) \sim (E,\mathit{len}(E'))$.

\subsection{Syntax}
\label{sec:syntax}
A formula $\psi$ takes the form: 
\begin{equation*}
  \begin{array}[t]{r@{\;::=\;}l}
    \psi & D_G (\varphi)  \; | \;  D_G (\psi) \; | \; \psi \mathrel{\land} \psi \; | \;
\lnot \psi \\[\jot]
    \varphi & p \; | \; \varphi \wedge \varphi \; | \; \neg \varphi
\; | \; \varphi S \varphi \; | \; \varphi U \varphi \; | \; \forall x ( \varphi)
  \end{array}
\end{equation*}

% \[  \psi ::= D_G (\varphi) \; | \; \psi \mathrel{\land} \psi \; | \;
% \lnot \psi \]
% \[ \varphi ::= p \; | \; \varphi \wedge \varphi  \; | \; \neg \varphi
% \; | \; \varphi S \varphi \; | \; \varphi U \varphi \;  | \; | \; \forall x ( \varphi) \]%  \; | \; \top\]%\] %\; | \; x : \tau \]  \mbox{call}(t,a,m,v) \; | \; 
% %with variables $t,a,v$, metavariable $x$, and $G \subseteq \textsc{Agents}$.
with $G \subseteq \textsc{Threads} \uplus \{ \mathit{obs} \}$ and $p \in \textsc{Predicates}$, which we instantiate for each of our characterizations.
The logic provides the temporal modalities $\varphi S \psi$
representing the fact that since $\psi$ occurred, $\varphi$ holds and
the modality $\varphi U \psi$ representing the fact that until $\psi$ occurrs, $\varphi$
holds. Additionally, it provides the \emph{distributed knowledge} modality $D_G$ and first order quantification.
% \begin{itemize}
% \item The temporal modality $\circleddash$, where $\circleddash \varphi$ means: $\varphi$ holds in the last state.
% \item The temporal modality $S$, where $\varphi S \psi$ means: Since $\psi$ occurred, $\varphi$ holds.
% \item The temporal modality $U$, where $\varphi U \psi$ means: Until $\psi$ occurrs, $\varphi$ holds.
% \item The epistemic modality $D_G$ representing \emph{distributed knowledge} among a group of threads $G$. A fact $\varphi$ is distributed knowledge among the threads in $G$ if the threads know $\varphi$ when they share everything they know.
% \item First order quantification.
% %\item The type judgement $x: \tau$, where $\tau := \textsc{T } | \textsc{ A } | \textsc{ V}$ are types for threads, addresses, and values respectively.\\
% \end{itemize}
Let $\Phi$ denote the set of all formulae in the logic.
\subsection{Semantics}
\label{def:semantics}
We now define the satisfaction relation $\models \;  \subseteq
(\mathcal{E}^\omega \times (\mathbb{N} \uplus \omega)) \times \Phi$. We let:
\begin{equation*}
\begin{array}[t]{@{}l@{\models}l@{\;\text{:iff}\;}l@{}}
(E,i)&\varphi \wedge \psi  & (E,i) \models \varphi$ and $(E,i) \models \psi\\
(E,i)&\neg \varphi &  \text{not} \; (E,i) \models \varphi \\
\end{array}
\end{equation*}
We define the temporal modalities by:
\begin{equation*}
\begin{array}[t]{@{}l@{\models}l@{\;\text{:iff}\;}l@{}}
%(E,i) & \circleddash \varphi & 
%\begin{array}[t]{ll}
% \; i>0 \text{ and } (E,i-1) \models \varphi \\
%\end{array}\\
(E,i) & \varphi S \psi & 
\begin{array}[t]{ll}
\text{ there is } j \leq i \text{ s.t. } (E,j) \models \psi \text{ and }\\
\text{ for all } j < k \leq i: (E,k)\models \varphi\\
\end{array}\\
(E,i) &  \varphi U \psi & 
\begin{array}[t]{ll}
\text{ there is } j \leq i \text{ s.t. } (E,j) \models \psi \text{ and} \\
\text{ for all } 1 \leq k < j: (E,k) \models \varphi
\end{array}
\end{array}
\end{equation*}
We define distributed knowledge as: 
$(E,i) \models D_G (\varphi)$ :iff for all $(E',i')$: if $(E,i) \sim_{G}
(E',i')$ then $(E',i') \models \varphi$, with $G \subseteq
\textsc{Threads} \uplus \{ \mathit{obs} \}$.
% The definition for the knowledge modality $D_G$ could be paraphrased as follows: Whenever the threads in a group $G$ observe a trace $E \downharpoonright i$ and there is a trace $E' \downharpoonright i'$ that they cannot distinguish from $E \downharpoonright i$, they cannot be sure whether it was $E \downharpoonright i$ or $E' \downharpoonright i'$ they saw. 
% If, however, a formula $\varphi$ holds for all the traces that they might have possibly observed, they know that $\varphi$ holds.
% E.g. for:
% \begin{center}
% $E := (t_2,\mathit{ld}(2)) (t_1,\mathit{st}(1))$\\
% $E' := (t_1,\mathit{st}(1))(t_2,\mathit{ld}(2))$\\
% \end{center}
% We have $E \sim_\textsc{Threads} E'$ and $E \sim_\textsc{Threads} E$,
% and by the definition of $D$ we get $E \models D_\textsc{Threads} \neg \mathit{correct}$, as $2$ was not stored.
% %The dual of knowledge is epistemic possibility, denoted by $\neg D_G \neg \varphi$. If there is at least one trace which the group $G$ cannot distinguish from the actual trace and on which $\varphi$ holds, group $G$ considers it possible that $\varphi$ is true. 
Let $\textsc{D}$ be the domain of quantification. We define first-order quantification:
$(E,i) \models \forall x(\varphi)$ :iff for all $d \in \textsc{D}: (E,i) \models \varphi [d/x]$.
By $\varphi [d/x]$, we denote the term $\varphi$ with all occurrences
of $x$ replaced by $d$. We define $\textsc{D}$ as the disjoint union
of all quantities used in the definition of a condition.
We write $E \models \varphi$ as an abbreviation for $(E,\mathit{len}(E)) \models \varphi$.
\paragraph{Additional Definitions}
\label{app:macro} For convenience, we define the following standard operators in terms
of our above definitions: $\varphi \vee \psi:= \neg (\neg \varphi \wedge
\neg \psi)$, $\varphi \rightarrow \psi := \neg \varphi \vee \psi$,
$\top := (p \vee \neg p)$ for some atomic predicate $p$,
$\diamondminus \varphi := \top S \varphi$ ("once $\varphi$"),
$\boxminus \varphi := \neg \diamondminus \neg \varphi$ ("so far
$\varphi$"), $\Diamond \varphi := \top U \varphi$ (``eventually
$\varphi$''), $\Box \varphi := \neg \Diamond \neg \varphi$ (``always
$\varphi$''), $\varphi W \psi := \varphi U \psi \vee \Box \varphi$
(``weak until''), $\exists x(\varphi) := \neg \forall x(\neg
\varphi)$. % We now present our characterizations of the consistency
           % conditions in mor detail.

\section{Sequential Consistency}
\label{sec:seqCons}
We present a trace-based definition of sequential consistency (cf.,
\cite{Attiya94}) and prove our logical characterization
equivalent. Our definition of sequential consistency generalizes the original
definition \cite{Lamport79} by allowing non-sequential specifications. 
\begin{definition}[Sequential Consistency]
\label{def:seq_cons}
Let $\textsc{Spec} \subseteq \mathcal{E}^\ast$ be a specification of
the shared data-structure.
A trace $E \downharpoonright i$ is sequentially consistent
seqCons$(E,i)$ if and only if there is $(E',i') \in \mathcal{E}^\omega
\times \mathbb{N}$ s.t.
 for all $t \in \textsc{Threads}$:
\begin{center}
 $(E \downharpoonright i) \downarrow
 t$ = $(E' \downharpoonright i') \downarrow t$ and $E'
 \downharpoonright i' \in \textsc{Spec}$
\end{center} 
\end{definition}
%E.g. the trace:
%\begin{center}
%$E:=(t_2, \mathit{call} \;  \mathit{ld}()) (t_1, \mathit{call} \;  \mathit{st}(1)) (t_2, \mathit{ret} \; \mathit{ld}(1))$  $(t_1,\mathit{ret} \;  \mathit{st}(\mathit{TRUE}))$
%\end{center}
%\subsection{Logical Characterization}
\paragraph*{Basic Predicates} For our logical characterization, we define the predicate $\mathit{correct}$ representing the fact that a trace meets the specification:
\begin{center}
 $(E,i) \models \mathit{correct}$ :iff 
  $E \downharpoonright i \in $\textsc{ Spec}\\
\end{center}
%In our characterizations of TSO (Section \ref{sec:TSO})  and eventual
%consistency (Section \ref{sec:evCons}), we provide two concrete instances of $\specCorrect$
%in temporal logic.
\begin{theorem}[Logical Characterization Sequential Consistency]
\label{thm:seq_cons}
A trace~$E \downharpoonright i $ is sequentially consistent if and
only if the threads do not jointly know that it is incorrect: 
%\begin{center}
$\mathit{seqCons}(E,i)$ iff $(E,i) \models \neg D_\textsc{Threads}(
\neg correct)$.
%\end{center}
\end{theorem}
\begin{proof}
By expanding the definitions in \ref{def:semantics}.
\end{proof}

\section{Eventual Consistency}
\label{sec:evCons}
%In this section we present our logical characterization of eventual consistency.
%  Eventual consistency \cite{evCons} is a correctness conditions
% empolyed for distributed database systems, were a weaker consistency
% requirement than sequential consistency is needed. In such systems,
% threads (clients) keep local copies (revisions) of the database. On their local copies,
%  threads issue updates and queries.  From time to time, the threads
%  reconnect to the database system by a commit operation. Upon commit,
% the threads update their local copy, i.e., they forward the state of
%  their revision and receive other threads' updates. The updates they
%  forward are however not instantaneously available to all other
%  threads. On their way through the network, the updates may be delayed
%  so that other threads receive them later -- or never at all.  
We define the set of actions for eventual consistency as: 
\begin{equation*}
\mathcal{A} \ni \mathit{act}:= \;  \mathit{qu}(\mathit{id},q,r) \;
\; | \; \; \mathit{up}(\mathit{id},u) \;
\; | \; \; \mathit{com}(\mathit{id}) \;
\; | \; \;  \mathit{fwd}(t,t',\mathit{id}).
\end{equation*}
% \begin{equation*}
% \begin{array}{rll}
% \mathcal{A} \ni \mathit{act}:= & \mathit{qu}(\mathit{id},q,r) \hspace{10ex} & \text{``query''}\\
% & | \; \mathit{up}(\mathit{id},u) & \text{``update''}\\
% & | \;  \mathit{com}(\mathit{id}) & \text{``commit''}\\
% & | \; \mathit{fwd}(t,t',\mathit{id}) & \text{``forward"}\\ 
% \end{array} 
% \end{equation*}
Threads may pose a query ($\mathit{qu}$) $q \in \textsc{Queries}$ with result $r \in
\textsc{Values}$, issue an update ($\mathit{up}$) $u \in
\textsc{Updates}$, or commit ($\mathit{com}$)
their local changes. Queries, updates and commits get assigned a
revision-id $\mathit{id} \in \textsc{Identifiers}$, representing the
current state of the local database copy. We assume that if a thread
commits, the committed revision id matches the revision id of the
previous queries and updates, and that thread-revision-id pairs~$(t,\mathit{id})$ are unique. Again, this is no restriction. To fulfill the requirement, the threads can just increment their local revision id whenever they commit.
As updates may get lost in the network, we represent by $ \mathit{fwd}(t,t',\mathit{id})$ the successful forwarding of the updates made by thread $t$ in revision $id$ to thread $t'$.
%We define eventual consistency on sets of events, that we call histories.% Let the set of histories $\mathcal{H} := \{ (t, \_) \} - \{ (t, \mathit{fwd}(\_))\}$, i.e. histories only include events caused by threads and do not include $\mathit{fwd}$ events. 

\paragraph*{Preliminaries}
We let let $\mathit{set}(E)= \{ e \in E \}$, i.e., the set of events
in trace E. On a fixed trace $E$, we define the program order
$\prec_{p}$ as $e \prec_{p} e'$ :iff if there is $t$ such that
$\mathit{pos}(e,E \downarrow t) < \mathit{pos}(e',E \downarrow t)$.
We let "$\_$" represent irrelevant, existential quantification.
Let $e \equiv_t e'$ if and only if there is $id \in \textsc{Identifiers}$ such that $e=(t,\_(\mathit{id},\_) )$ and $e'=(t,\_(\mathit{id},\_))$, i.e., if the events belong to the same revision of thread $t$. A relation $\preceq$ factors over $\equiv_t$ if $x \preceq y$, $x \equiv_t x'$ and $y \equiv_t y'$ imply $x' \preceq y'$. Updates are interpreted in terms of states, i.e., we assume there is an interpretation function $u^{\#}: \textsc{States} \to \textsc{States}$, for each $u \in \textsc{Updates}$, and a designated initial state $s_0 \in \textsc{States}$. For each query $q \in \textsc{Queries}$, there is an interpretation function $q^{\#}: \textsc{States} \to \textsc{Values}$.
For a finite set of events $E_S$, a total order~$\prec$ over the events in $E_S$, and a state $s$ we let $\mathit{apply}(E_s,\prec,s)$ be the result of applying all updates in $E_s$ to $s$, in the order specified by $\prec$.
%\paragraph*{Definition (Eventual Consistency)}
\begin{definition}[Eventual Consistency]
We use the definition presented in \cite{evCons}.
A trace~$E \in \mathcal{E}^\infty$ is eventually consistent ($\mathit{evCons}(E)$) if and only if there exist a partial order $\prec_{v}$ (visibility order), and a total order $\prec_{a}$ (arbitration order) on the events in $\mathit{set}(E)$ such that:
\begin{itemize}
\item $\prec_{v}  \subseteq \prec_{a}$ (arbitration extends visibility). 
\item $\prec_{p} \subseteq \prec_{v}$ (visibility is compatible with program-order).
\item for each $e_q = (t,\mathit{qu}(\mathit{id},q,r)) \in E$, we have r = $\mathit{apply}(\{e \; | \; e \prec_{v}  e_q \},\prec_{a},s_0)$ (consistent query results).
\item $\prec_{a}$ and $\prec_{v}$ factor over $\equiv_t $ ( atomic revisions).
\item if $(t, \mathit{com}(\mathit{id})) \not\in E$ and $(t,\_(\mathit{id},\_)) \prec_{v} (t',\_)$ then $t=t'$ (uncommitted updates).
%if $(t,\mathit{com},\mathit{id}) \notin E$, $\mathit{trans}(e)=(t,id)$, $\mathit{trans}(e')=(t',id')$ and $e \prec_{v} e'$ then $t=t'$.
\item if $e= (t,\mathit{com}(\mathit{id})) \in \mathit{E}$ then there are only finitely many $e' := (t',\mathit{com}(\mathit{id}'))$ such that $e' \in E$ and $e \not\prec_{v} e'$ (eventual visibility).
\end{itemize} 
\end{definition}
% In the example from the introduction:
% \begin{equation*}
% \begin{array}{rl}
% E_{10} :=&(t_1, \mathit{up}(0,x:=0)) \;  (t_1,\mathit{com}(0))\; (t_1,\mathit{up}(1,x:=1))\\ [\jot]
%  & (t_1,\mathit{com}(1))  \; (t_2,\mathit{qu}(0,x,0)) (t_2,\mathit{com}(0)) \\ [\jot]
% &  (t_2,\mathit{qu}(1,x,1)) \\  [\jot]
% \end{array}
% \end{equation*}
% we can fix the order of events in the trace as arbitration order and
% exclude edges from events from $(t_1,1)$ to $(t_2,0)$ from the
% visibility order.

%E.g., the trace:
% \begin{center}
% $E_9 := (t_1,\mathit{up}(1,x,1)) (t_1,\mathit{com}(1)) (t_2,\mathit{qu}(2,x,0))$ $(t_2,\mathit{qu}(2,x,1))$\\
% \end{center}
% is eventually consistent, because we can take the order of the trace as arbitration order and exclude $t_1$'s update and $t_2$'s first query from the visibility order i.e., $(t_1,\mathit{up}(1,x,1)) \not \prec_v (t_2,\mathit{qu}(2,x,0))$.
\subsection{Logical Characterization}
%In this section we present our logical characterization of eventual consistency.
%  In our characterization, a trace is eventually consistent if the
%  threads do not know that the trace is incorrect with respect to a
%  specification $\specEVC$ that requires threads to keep valid logs. A
%  valid log includes all the information a thread received via updates 
% from other threads in an order that is consistent with the other threads' logs.
% For a log to be valid for a thread, it must:
% \begin{enumerate}
% \item Include all the updates that the thread knows of i.e. its own updates and the ones it received from other threads.
% \item Provide a consistent view on the data, i.e. it must agree with the other threads logs on the order of common updates.
% \end{enumerate}
% Since logs must include all operations that the thread knows, once a thread receives an update, it must insert the update into its logs in a way that is consistent with other threads logs.
% Whenever a thread issues a query, the query must be correct with respect to the current state of its logs. 
\paragraph*{Basic Predicates}
%We present the semantics of our predicates in Table \ref{tab:satEv}.
 We represent queries and updates by predicates
 $\mathit{query}(t,q,r,\mathit{id})$ and
 $\mathit{update}(t,u,\mathit{id})$, representing $t \in
 \textsc{Threads}$, issuing $q \in \textsc{Queries}$ with result $r
 \in \textsc{Values}$ on revision $id \in \textsc{Identifiers}$, and
 $t$ performing $u \in \textsc{Updates}$ on revision $\mathit{id}$,
 respectively.  As threads work on their local copies, revision ids
 mark the version of data the threads work with. We represent commits
 by the predicate $\mathit{commit}(t,\mathit{id})$, representing $t$
 committing its state in revision $id$. After performing a commit, a
 new revision id is assigned. We define:
\begin{equation*}
  \begin{array}[t]{@{}l@{\models}l@{\;\text{:iff}\;}l@{}}
(E,i) & \mathit{query}(t,q,r,\mathit{id}) & E@i=(t,\mathit{qu}(\mathit{id},q,r))\\[\jot]
(E,i) &   \mathit{update}(t,u,\mathit{id})& E@i=(t,\mathit{up}(\mathit{id},u))\\[\jot]
(E,i) &   \mathit{commit}(t,\mathit{id}) & E@i=(t, \mathit{com}(\mathit{id}))
\end{array}
\end{equation*}

We let $\mathit{query}(t,q,r) \; := \; \exists
\mathit{id}(\mathit{query}(t,q,r,\mathit{id}))$. Upon commit, a thread forwards all the information from its local repository to the database system and receives updates from other threads. Committed updates may however be delayed or lost by the network. By predicate forward$(t,t',id)$ we mark the event that the environment forwarded the updates $t$ performed in revision $id$ to $t'$. We let: 
%\begin{center}
$(E,i) \models  \mathit{forward}(t,t',\mathit{id})$ :iff $E@i=(\mathit{env},\mathit{fwd}(t,t',\mathit{id}))$.
%\end{center}
Eventual Consistency, requires all threads to keep valid logs. Logs
are finite sequences of actions, i.e., $\mathcal{L} \in A^\ast$. We represent log validity by the formula:
%\begin{center}
$\mathcal{L} \; \mathit{validLog}  \; t \; := \; \forall a(t \; k_{\mathit{log}} \;  a \leftrightarrow a \; \mathit{in} \; \mathcal{L}) \; \wedge \mathit{consistent}(\mathcal{L})$.
%\end{center}
That is, to be a valid log for thread~$t$, log~$\mathcal{L}$ must contain exactly the actions that $t$ knows of and these actions must be arranged in an order consistent with respect to the other threads logs. We let:
%\begin{center}
$(E,i) \models \; a \; \mathit{in} \; \mathcal{L} \; :$iff $a \in L$.
%\end{center}
%\paragraph*{Knowing About Actions}
The predicate $t \; k_{\mathit{log}} \; a$ represents the fact that
$t$ knows about action $a$. The predicate~$k_{\mathit{log}}$ represents individual
knowledge, i.e., knowledge in the sense of knowing about an action in contrast to knowing that a fact is true \cite{kramer10}. We let:
\begin{center}
\begin{tabular}{lll}
$(E,i) \models$&$t \; k_{\mathit{log}} \; a$ :iff  there is $j \leq i: (E@j= (t,a)$ or \\[\jot]
&($(E,j) \models \mathit{forward}(t',t,id)$ and there is $l<j: \; (E,l) \models \mathit{commit}(t',id)$\\[\jot]
&and $(E,l) \models t' \; k_{log} \; a$))  
\end{tabular}
\end{center}
That is, threads know an action if they performed it themselves, or
they received an update containing it. Upon commits, threads pass on all
actions they know about.
%\paragraph*{Consistent Logs}
A log $\mathcal{L}$ is consistent if the actions in the log occur in
the same order as the actions in the real trace. This means the sequence of actions in $\mathcal{L}$ must be a subsequence of the actions in the real trace. A sequence $a= a_1 a_2 \ldots a_n$ is a subsequence of a sequence $b=b_1 b_2 \ldots b_m$ ($a \preceq b$), if and only if there exist $1 \leq i_1 < i_2 <  \ldots < i_n \leq m$ such that for all $1 \leq j \leq n: a_j = b_{i_j}$. We project a sequence of events to a sequence of actions by the function $\mathit{act}: \mathcal{E}^\ast \to \mathcal{A}^\ast$, such that
%\begin{equation*}
$\mathit{act}((t_1,a_1) (t_2,a_2) \ldots (t_n,a_n)) = a_1a_2 \ldots  a_n$.
% \end{equation*}
We define:
%\begin{center}
$(E,i) \models \mathit{consistent}(L):$iff $ \mathcal{L} \preceq \mathit{act}(E \downharpoonright i)$.\\
%\end{center}
\paragraph{Query Results}
All queries that threads issue must return the correct result with
respect to the logged operations. That is, the query's result must
match the result the query would yield when issued on a database that
performed all the updates in the log. We represent the fact that
query~$q$ would yield result $r$ on log $\mathcal{L}$ by the predicate result$(q,\mathcal{L},r)$.
We define the order of actions in a log $\mathcal{L}$ by the relation $<_\mathcal{L}$. We let $a <_\mathcal{L} a'$ :iff $pos(a,\mathcal{L}) < pos(a',\mathcal{L}) < \omega$.
We define: 
%\begin{center}
$(E,i) \models \mathit{result}(q,L,r):$iff $r=q^\#(\mathit{apply}(set(\mathcal{L}),<_\mathcal{L},s_0))$.
%\end{center}
\paragraph*{Network Assumptions}
We pose additional requirements on the network:
%\begin{itemize}
Updates in the same revision must be sent as atomic bundles
  ($\mathit{atomicTrans}$).
Only commited updates can be forwarded ($\mathit{fwd}$).
Active threads must eventually receive all committed update
  ($\mathit{alive}$).
%\end{itemize}
We formalize them in Appendix \ref{app:EVC}.
We represent $\mathit{correctEVC}$ by the formula: 
\begin{equation*}
  \mathit{correctEVC} := 
  \begin{array}[t]{@{}l@{}}
    \forall t \forall q \forall r \\[\jot]
  \begin{array}[t]{@{}l@{}}
	 ( \boxminus( \mathit{query}(t,q,r) \rightarrow
	 \exists \mathcal{L}( \mathcal{L} \; \mathit{validLog} \; t \;  \wedge  \mathit{result}(q,\mathcal{L},r)))) \\ [\jot]
         \wedge \mathit{atomicTrans} \wedge \mathit{alive} \wedge \mathit{fwd} 
  \end{array}
\end{array}
\end{equation*}
% That is, whenever a thread performs a query, the thread must be able to provide a witness log $\mathcal{L}$ such that:
% \begin{itemize}
% \item $\mathcal{L}$ is valid, 
% \item the query $q$ executed on a local database copy that performed all the actions logged in $\mathcal{L}$ yields r.
% \end{itemize}
% To be correctEVC, a trace additionally needs to satisfy the following further requirements:
% \begin{itemize}
% \item Updates made in the same revision must be sent bundled as indivisible transactions ($\mathit{atomicTrans}$).
% \item A thread that makes progress, i.e. that commits infinitely often
%   must eventually receive all committed updates ($\mathit{alive}$).
% \item Only committed revisions can be forwarded ($\mathit{fwd}$).
% \end{itemize}

%\paragraph*{Theorem (Logical Characterization Eventual Consistency)}
\begin{theorem}[Logical Characterization Eventual
  Consistency]
\label{thm:evCons}
A trace is eventually constent if and only if the threads do not know
that it violates $\specEVC$. For all traces $E \in \mathcal{E}^\infty$:

\begin{center}
$evCons(E)$ if and only if $E \models \neg D_{\textsc{Threads}} \neg(\mathit{correctEVC})$.
\end{center}
% Consider traces $E_9$ and $E_{10}$:
% \begin{equation*}
% \begin{array}[t]{@{}l@{}l@{}}
% E_9 :=& (t_1,\mathit{up}(1,x,1)) (t,\mathit{com}(1)) (t_2,\mathit{qu}(2,x,0))\\ 
% & (t_2,\mathit{qu}(2,x,1))\\
% E_{10} := & (t_1,\mathit{up}(1,x,1)) (t_1,\mathit{com}(1)) (t_2,\mathit{qu}(2,x,0)) \\
% & (d_{\mathit{sys}},\mathit{fwd}(t_2,t_1,1)) (t_2,\mathit{qu}(2,x,1))\\
% \end{array}
% \end{equation*}
% Trace $E_9$ is eventually consistent, because we can find $E_{10}$,
% such that $E_9 \sim_\textsc{Threads} E_{10}$ and $E_{10} \models
% \mathit{correctEVC}$. I.e. for its first query, $t_2$ can present
% the valid log $\mathcal{L}:= (t_2,\mathit{qu}(1,x,0))$, and for its
% second query the log $\mathcal{L}:= (t_1,\mathit{up}(1,x,1))
% (t_1,\mathit{com}(1))$ $(t_2,\mathit{qu}(1,x,0))$
% $(t_2,\mathit{qu}(1,x,1))$, and the results of the queries match the
% logs.
\end{theorem}
%\begin{proof}
%We give a proof sketch in Appendix \ref{app:Proof_EVC}.

% \ifthenelse{\equal{\isTechReport}{false}}{
% We only give a proof sketch and provide the full proof in \cite{techrep}.
% }{
% We only give a proof sketch and provide the full proof in Appendix \ref{proof:evCons}.
% }
% For the left to right direction, for each revision $(t,\mathit{com}(id))$, we insert $\mathit{fwd}(t,t,',id)$ events, directly before the minimal wrt. $\prec_a$ revision $(t',\mathit{com}(id'))$ such that $(t,\mathit{com}(id)) \prec_v (t',\mathit{com}(id)')$. We then need to show that the resulting trace satisfies $\mathit{correctEVC}$. For the right to left direction, a revision is visible whenever it is in the thread's log. Transitivity of the visibility order is ensured by our definition of $k_{log}$.

%\end{proof}

\section{Linearizability}
\label{sec:lin}
Linearizability refines sequential consistency by guaranteeing that each method call takes its effect at exactly one point between its invocation and its return.

%\subsection{Definition}
For our definition of linearizability, we follow \cite{gotsman11}. As for sequential consistency, our definition generalizes the original notion \cite{Herlihy87,Herlihy90} by allowing non-sequential specifications. We define the real-time precedence order $\preceq_{real}$ $\subseteq (\mathcal{E}^\omega \times \mathbb{N})^2$:
%\item $(E,i)$ presProgOrder $(E',i')$ :iff for all $t \in \textsc{Threads}: (E \downharpoonright i)\downarrow t = (E' \downharpoonright i')\downarrow t$ \\
$(E,i)$ $\preceq_{real}$ $(E',i')$ :iff there is a bijection $\pi: \{1,\ldots,i\} \rightarrow  \{1,\ldots,i'\} $  s.t
for all $j \in \mathbb{N}$ such that $j \leq i:   E@j = E'@\pi(j)$, i.e., $E'$ is a permutation of $E$, and
for all $j,k \in \mathbb{N}$  such that $j < k \leq i:$ if $E@j \in \textsc{Ret} \mbox{ and } E@k \in \textsc{Call} \mbox{ then } \pi(j) < \pi(k)$, i.e., when permuting the events in $E$, calls are never pulled before returns.
\begin{definition}[Linearizability]
\label{def:lin}
A trace $(E,i)$ is linearizable ($\mathit{lin}(E,i)$) if and only if there is $(E',i')  \in \mathcal{E}^\omega \times \mathbb{N}$ such that 
%\begin{itemize}
(1) for all $t \in \textsc{Threads}$: $(E \downharpoonright i) \downarrow t$ = $(E' \downharpoonright i') \downarrow t$
(2)  $(E,i) \preceq_{real} (E',i')$ and
(3) $E' \downharpoonright i' \in \textsc{Spec}$.
%\end{itemize}
\end{definition}
To characterize linearizability, we make the assumption that each event occurs only once in a trace. This is not a restriction as we could add a unique time-stamp or a sequence number to each event.
\begin{theorem}[Logical Characterization Linearizability]
\label{thm:lin}
A trace $E \downharpoonright i  \in \mathcal{E}^\ast$ is linearizable if and only if the threads together with the observer do not know that it is incorrect: 
\begin{center}
$\mathit{lin}(E,i)$ iff $(E,i) \models \neg D_{\textsc{Threads} \uplus \{ \mathit{obs} \} } \neg correct$.
\end{center}
\end{theorem}

%\begin{proof}
%We provide a proof sketch in Appendix \ref{app:lin_proof}.
%\end{proof}
% \subsection{Detecting Linearizability}
% As opposed to sequential consistency, for linearizability threads together with the observer can only detect violations of linearizability but not adherence to the condition. That is, whenever a trace is not linearizable the threads and the observer know that it is not; however, there are linearizable traces on which the threads and the observer do not know that the traces are linearizable.
% Let $\mathit{Lin} := \neg D_{\textsc{Threads} \uplus \{\mathit{obs} \}} \neg \mathit{correct}$. 
% \paragraph*{Theorem (Detection Linearizability)}
% There are linearizable traces on which the threads together with the observer do not know that they are linearizable. There is $E \downharpoonright i \in \mathcal{E}^\ast$ such that 
% \begin{equation*}
%   \begin{array}[t]{@{}l@{}}
%   (E,i) \models \mathit{Lin} \land \neg D_\mathit{Threads \uplus \{ \mathit{obs}\}}(\mathit{Lin})
%   \end{array}
% \end{equation*} 
% As in sequential consistency, the threads together with the observer can spot when a trace is not linearizable.
% \begin{equation*}
% \models  \neg \mathit{Lin} \leftrightarrow D_\mathit{Threads  \uplus \{ \mathit{obs}\}}(\neg\mathit{Lin})
% \end{equation*}

% \begin{proof}
% As $\sim_\mathit{obs}$ is a partial order, axiom (5) for negative introspection is not a validity. That means its negation is satisfiable. The proof for the second claim is analogue to the case of sequential consistency.
% \end{proof}

\section{Knowledge about consistency}
\label{sec:detect}
%\paragraph{Detecting Sequential Consistency}
%In this section we show that for sequential consistency the threads
%can detect both violations and adherence to the correctness condition.
We write $\models \varphi$ as an abbreviation for:
for all $E \in \mathcal{E}^\infty$: $E\models \varphi$. 
Let $\seqCons := \neg D_{\textsc{Threads}} (\neg \mathit{correct})$. 
\begin{theorem}[Detection Sequential Consistency]
\label{thm:dec_seqs}
Threads can decide whether a trace is sequentially consistent or not:
%\begin{equation*}
%  \begin{array}[t]{@{}l@{}}
$    \models (\seqCons \leftrightarrow D_\mathit{Threads}(\seqCons)) \mathrel{\land}
    (\neg\seqCons \leftrightarrow D_\mathit{Threads}(\neg\seqCons))$.
%  \end{array}
%\end{equation*} 
\end{theorem}
%\paragraph{Detecting Linearizability}
%As opposed to sequential consistency, for linearizability threads together with the observer can only detect violations of linearizability but not adherence to the condition. That is, whenever a trace is not linearizable the threads and the observer know that it is not; however, there are linearizable traces on which the threads and the observer do not know that the traces are linearizable.
Let $\mathit{Lin} := \neg D_{\textsc{Threads} \uplus \{\mathit{obs} \}} \neg \mathit{correct}$. 
\begin{theorem}[Detection Linearizability]
\label{thm:dec_lin}
There is $E \in \mathcal{E}^\infty$ such that 
$(E,i) \models \mathit{Lin} \land \neg D_\mathit{Threads \uplus \{
  \mathit{obs}\}} (\Lin)$.
As in sequential consistency, the threads together with the observer can spot if a trace is not linearizable:
$\models  \neg \mathit{Lin} \leftrightarrow D_\mathit{Threads  \uplus \{ \mathit{obs}\}}(\neg\mathit{Lin})$.
\end{theorem}

\section{Related Work}
\label{sec:rel}

The only applications of epistemic logic to concurrent computations that we are aware of are a logical characterization of wait-free
computations by Hirai \cite{Hirai10} and a knowledge based analysis of cache-coherence by Baukus et al. \cite{Baukus04}. %There are a number of special-purpose logics for communicating agents (e.g. \cite{raman96,Chandy86}). %A characterization is such a logic would however come at the cost of having to build verification tools from scratch. 
 \bibliographystyle{abbrv}
% %\bibliographystyle{nature}
 \bibliography{Literatur}

\appendix

\section{Additional Definitions Eventual Consistency}
% \subsection{TSO}
% \label{app:TSO}
% We formalize balancedness as:
% \begin{equation*}
% \begin{array}{rl}
% \forall t \forall a \forall v \forall t' \forall a' \\[\jot]
% &(\boxminus (\mathit{flush}(t,a,v) \rightarrow \mathit{stored}(t,a,v))
% \; \; \wedge \\[\jot]
% & (\boxminus (\mathit{store}(t,a,v) \rightarrow \circleddash
%  \neg \mathit{stored}(t',a',v))) \; \; \wedge \\[\jot]
% & (\boxminus( \mathit{flush}(t,a,v) \rightarrow \circleddash \neg \mathit{flushed}(t'a'v))))
% \end{array}
% \end{equation*}

% % \begin{tabular}{l r l}
% % balanced :=  $\forall t \forall a \forall v \forall t' \forall a' ($\\%$flushed$(t,a,v) \rightarrow$ stored$(t,a,v) \; \wedge $ \\
% % %&$\neg$ store$(t,a,v) B$ flush$(t,a,v) \; \wedge$\\   
 
% % $\wedge \; (\boxminus( \mathit{flush}(t,a,v) \rightarrow \circleddash \neg \mathit{flushed}(t'a'v)))$).\\
% % \end{tabular}
% % \end{center}
\label{app:EVC}
We define the helper predicate:
\begin{center}
$\mathit{rev}(t,id)$ := $\exists q \exists r (\mathit{query}(t,q,r,id)) \vee \exists u(\mathit{update}(t,u,\mathit{id}))$
\end{center}
representing the fact, that the current action belongs to revision $id$ of thread $t$.  We specify the requirements that updates made in the same revision must be sent bundled as indivisible transactions by the formula :
 \begin{equation*}
 \mathit{atomicTrans} :=
  \forall t \forall id
    (\boxminus (\mathit{rev}(t,\mathit{id})\rightarrow
	 \WeakUntil{\mathit{rev}(t,\mathit{id})} {\mathit{commit}(t,\mathit{id})}  ))
 \end{equation*}
That is, queries and updates from revision $id$ are only followed by
other queries and updates from the same revision, or a commit. We enforce that only committed revisions can be forwarded by:
\begin{equation*}
 \mathit{fwd} :=
  \forall t \forall t' \forall id
    (\boxminus (\mathit{fwd}(t,t',\mathit{id})\rightarrow
	 \Box (\neg \mathit{commit}(t,\mathit{id}) ) ))
 \end{equation*}
Threads that makes progress, i.e. that commit infinitely often must eventually receive all committed updates. We formalize this as:
\begin{equation*}
\mathit{alive} :=
\begin{array}[t]{ll}
  \forall t \forall t' \forall \mathit{id}  \\
( \boxminus (\mathit{commit}(t,\mathit{id}) \wedge \Box \Diamond (\exists \mathit{id}'( \mathit{commit}(t',id')) ) \rightarrow \\[\jot] 
%\\[\jot]
\Diamond \mathit{forward}(t,t',\mathit{id}))) \\[\jot] 
  \end{array}
\end{equation*}
That is, if thread~$t'$ commits infinitely often, eventually, the database system must manage to forward all the committed updates.
\section{Proofs}
\subsection{Linearizability (Theorem \ref{thm:lin})}
\label{app:lin_proof}
A trace $E \downharpoonright i  \in \mathcal{E}^\ast$ is linearizable if and only if the threads together with the observer do not know that it is incorrect: 
\begin{center}
$\mathit{lin}(E,i)$ iff $(E,i) \models \neg D_{\textsc{Threads} \uplus \{ \mathit{obs} \} } \neg correct$
\end{center}

We formalize our assumption that each event occurs at most once in a trace:
\begin{center}
\begin{tabular}{l l}
unique := & for all E $\in \mathcal{E}^\omega$ and all $j,k \in \mathbb{N}$:\\  
&  if $E@j=E@k$ then $j=k$.\\
\end{tabular}
\end{center}
\label{sec:proofLin}
\begin{definition}[Eventset]
\label{def:setproj}
Let $\llbracket \cdot \rrbracket$ : $\mathcal{E}^\ast \to
\mathcal{P}(\mathcal{E})$ denote a function that transforms a trace into the set of events it contains, that is: \\
For all $E \in \mathcal{E}^\ast: \llbracket E \rrbracket = \{ e \; |
\; e \in E \}$.\\
\end{definition}
\begin{proposition}[Union of Thread-Eventsets]
\label{prop:disjoint}
For all $E \in \mathcal{E}^\ast$:  $\llbracket E \rrbracket = \uplus_{t}  \llbracket E \downarrow t \rrbracket$.
\end{proposition}
\begin{proof}
By induction on $\mathit{len}(E)$.\\
For $\mathit{len}(E)=0$, we have $E=\epsilon$ and  $\varnothing=\varnothing$.\\
Let $E \cdot E'$ denote the concatenation of traces $E$ and $E'$.
For $\mathit{len}(E) = n+1$, we have $E = e \cdot E'$ for some $e \in \mathcal{E}$, $E' \in \mathcal{E}^\ast$.\\
We get $\llbracket E' \rrbracket \uplus \{ e \} = \uplus_{t' \neq t}  \llbracket E' \downarrow t' \rrbracket \uplus  \llbracket E' \downarrow t \rrbracket \uplus \{ e \}$, for some t, and by the induction hypothesis: $\llbracket E' \rrbracket \uplus \{ e \} = \llbracket E' \rrbracket \uplus \{ e \}$.
\end{proof}
%\begin{proof}
%By induction on len(E).
%\begin{tabular}{ll}
%0: & then E=$\epsilon$. We have $\varnothing=\varnothing$.\\
%$n \rightarrow n+1$& $E=e \cdot E'$. We have  \\
%&$\llbracket E' \rrbracket \uplus \{ e \} = \uplus_{t' \neq t}  \llbracket E' \downarrow t' \rrbracket \uplus  \llbracket E' \downarrow t \rrbracket \uplus \{ e \}$\\
%& $\llbracket E' \rrbracket \uplus \{ e \} = \llbracket E' \rrbracket \uplus \{ e \}$
%\end{tabular}
%\end{proof}
\begin{lemma}
\label{lemma:card}
For all $(E,i),(E',i') \in \mathcal{E}^\omega \times \mathbb{N}$:  if unique and $(E,i)  \sim_{\textsc{Threads}} (E',i')$ then $i=i'$.
\end{lemma}
\begin{proof}
For a proof by contradiction we assume $i>i'$ without loss of generality. Then, by unique, there is an $e \in \mathcal{E}$ such that $e \in \llbracket E \downharpoonright i \rrbracket$ and $e \notin \llbracket E'  \downharpoonright i' \rrbracket$. By Proposition \ref{prop:disjoint}: $e \in \llbracket (E \downharpoonright i)\downarrow t \rrbracket$ for some $t \in \textsc{Threads}$ but $e \notin \llbracket (E' \downharpoonright i')\downarrow t \rrbracket$. But by  $(E,i) \sim_{\textsc{Threads}} (E',i')$, we have $(E \downharpoonright i) \downarrow t $ =  $(E' \downharpoonright i') \downarrow t$ and thus $\llbracket (E \downharpoonright i)\downarrow t \rrbracket $ =  $\llbracket (E' \downharpoonright i') \downarrow t \rrbracket$, from which we get the contradiction.\hfill$\Box$
\end{proof}
\begin{lemma}
\label{lemma:exist}
For all $(E,i),(E',i') \in \mathcal{E}^\omega \times \mathbb{N}$:  if $(E,i)  \sim_{\textsc{Threads}} (E',i')$ and $E@j=e$ for some $j \in \mathbb{N}$ such that $j \leq i$  then there is $j' \in \mathbb{N}$ such that $1 \leq j' \leq i'$ and $E'@j' = e$. 
\end{lemma}
\begin{proof}
Suppose $j \leq i$, $E@j=e$ and $(E,i)  \sim_{\textsc{Threads}}
(E',i')$. By Proposition \ref{prop:disjoint}, $e \in \llbracket (E
\downharpoonright i)\downarrow t \rrbracket$ for some $t \in
\textsc{Threads}$. Then because  $(E \downharpoonright i) \downarrow t
$ =  $(E' \downharpoonright i') \downarrow t$ we have $\llbracket (E
\downharpoonright i)\downarrow t \rrbracket $ =  $\llbracket (E'
\downharpoonright i') \downarrow t \rrbracket $ and thus by
proposition \ref{prop:disjoint}: $e \in \llbracket (E' \downharpoonright i') \rrbracket $ and thus by definition \ref{def:setproj}, $E'@j'=e$ for some j' with $ j' \leq  i'$.\hfill$\Box$
\end{proof}
\begin{lemma}[Existence of a Bijection]
\label{lemma:func}
For all $(E,i),(E',i') \in \mathcal{E}^\omega \times \mathbb{N}$:  if unique and $(E,i) \sim_{\textsc{Threads}} (E',i')$ then there exists a bijective function $\pi: \{1,\ldots,i\} \rightarrow  \{1,\ldots,i'\} $ and 
for all $j \in \mathbb{N}$ such that $j \leq i:   E@j = E'@\pi(j)$.
\end{lemma}
\begin{proof}
Let $j \in \mathbb{N}$ such that $j \leq i$ and $E@j=e$. By lemma \ref{lemma:exist} we know that $E'@j'= e$ for some $j' \in \mathbb{N}$ with $j' \leq i'$. We will now show that the mapping from j to j' is a function. Suppose there was $k' \neq j'$ with $k' \in \mathbb{N}$ such that $E'@k'=e$ and $k' \leq i'$. This cannot be, since by unique each event occurs at most once in each trace. Let us denote that mapping by $\pi$. We now need to show that $\pi$ is a bijection.
By lemma \ref{lemma:card}, $i=i'$ and we have $\pi: \{ 1 \ldots i \} \to \{ 1 \ldots i \}$. This means it suffices to show that $\pi$ is injective. Now for a contradiction suppose that for $j,k \in \mathbb{N}$ with $ j,k \leq i$: $E@j= e$ and $E@k=e'$ for some $e,e' \in \mathcal{E}$
with $j \neq k$ and  thus by unique $e \neq e'$. Now let $\pi(j)=j'$ and $\pi(k)=j'$ for some $j' \in \mathbb{N}$ with $1 \leq j' \leq i'$. Then $E'@j'=e=e'$, contradicting $e \neq e'$. 
\end{proof}
\begin{landscape}
\begin{table*}
\footnotesize
\caption{Equivalence proof for linearizability: ($\rightarrow$)}
\label{proof:ltr}
\begin{center}
\fbox{
\begin{tabular}{l l l}
& Let $\preceq_{lin} \; := \; \sim_{\textsc{Threads}} \cap \preceq_{real}$ and $\mathcal{A} := \textsc{Threads} \uplus
\{ obs \}$. \\
Show:&for all $(E,i) \in  \mathcal{E}^\omega \times \mathbb{N}$: \\
&if  lin (E,i) then there is $(E',i') \in (\mathcal{E}^\omega \times \mathbb{N}):$  s.t. $(E,i)\sim_{\mathcal{A}}(E',i')$ and $(E',i') \models correct$.\\
\\
1.& $(E,i) \in  \mathcal{E}^\omega \times \mathbb{N}$ \hfill hyp.\\
%2.& \tab for all E $\in \mathcal{E}^\omega$ and all $j,k \in \mathbb{N}$ if $E@j=E@k$ then $j=k$.  \hfill hyp.\\
2.& \tab \tab lin (E,i) \hfill hyp. \\
3.& \tab \tab $(E',i') \in \mathcal{E}^\omega \times \mathbb{N}$ and $(E,i)  \preceq_{lin} (E',i')$ and $(E',i') \models $ correct \hfill 2,def. lin\\ 
4.& \tab \tab \tab  $(e,e') \in obs(E , i)$ \hfill  hyp. \\
5.& \tab \tab \tab \tab $j,k \in \mathbb{N}$ and $E@j=e$ and $E@k=e'$ and $e \in \textsc{Ret}$ and $e' \in \textsc{Inv}$ and $j < k \leq i$ \hspace{0.5cm} \hfill 4,def.obs,def.pos\\
6.& \tab \tab \tab \tab $\pi: \{1,\ldots,i\} \rightarrow  \{1,\ldots,i'\}$ is a bijective function and $E'@\pi(j)=e$ and \\
& \tab \tab \tab \tab \tab \tab $E'@\pi(k)=e'$  and $\pi(j)<\pi(k)$ \hfill 3,5,def. $\preceq_{lin}$,def. $\preceq_{real}$ \\
7.& \tab \tab \tab \tab $\pi(j) < \pi(k) \leq i'$ \hfill 5,6,def. $\pi$ \\ 
8.& \tab \tab \tab $(e,e') \in obs(E', i')$ \hfill 5,6,7,def. obs \\
9.& \tab \tab \tab  $obs(E, i) \subseteq obs(E', i')$ \hfill 4,8,def. $\subseteq$ \\
10.& \tab \tab  \tab  $(E,i) \sim_{obs} (E',i')$ \hfill 9, def. $\sim_{obs}$ \\
11.& \tab \tab  \tab  $(E,i) \sim_{{\textsc{Threads}}} (E',i')$ \hfill 3, def. $\preceq_{lin}$\\
11.& \tab \tab  \tab $(E,i) \sim_{\textsc{A}} (E',i')$ \hfill 10,11,def. $ \sim_{\textsc{A}}$\\
12.&  \tab \tab there is $(E',i') \in (\mathcal{E}^\omega \times \mathbb{N})$  s.t. $(E,i)\sim_{\mathcal{A}}(E',i')$ and $(E',i') \models correct$ \hfill 3,11\\
13.& for all $(E,i) \in  \mathcal{E}^\omega \times \mathbb{N}$: \\
&if  lin (E,i) then there is $(E',i') \in (\mathcal{E}^\omega \times \mathbb{N}):$  s.t. $(E,i)\sim_{\mathcal{A}}(E',i')$ and $(E',i') \models correct$. \hfill 1,2,12\\
\end{tabular}
}
\end{center}
\end{table*}
\end{landscape}
\begin{landscape}
\begin{table*}
\footnotesize
\caption{Equivalence proof for linearizability: ($\leftarrow$)}
\label{proof:rtl}
\begin{center}
\fbox{
\begin{tabular}{l l l}
Show:& For all $(E,i) \in \mathcal{E}^\omega \times \mathbb{N}$:\\
&if unique then if there is $(E',i') \in \mathcal{E}^\omega \times \mathbb{N}:$  s.t. $(E,i)\sim_{\mathcal{A}}(E',i')$ and $(E',i') \models correct$  then lin (E,i). \\
\\
1.& $(E,i) \in  \mathcal{E}^\omega \times \mathbb{N}$ \hfill hyp.\\
2.& \tab for all E $\in \mathcal{E}^\omega$ and all $j,k \in \mathbb{N}$ if $E@j=E@k$ then $j=k$. \hfill hyp.\\
3.& \tab \tab $(E',i') \in \mathcal{E}^\omega \times \mathbb{N}$ and $(E,i)\sim_{\mathcal{A}}(E',i')$ and $(E',i') \models correct$ \hfill hyp.\\
4.& \tab \tab  $(E,i) \sim_{\textsc{Threads}}(E',i')$ \hfill 3,def. $\sim_{\mathcal{A}}$\\
\\
5.& \tab \tab $\pi: \{1,\ldots,i\} \rightarrow  \{1,\ldots,i'\}$ is a bijective function and \\
   & \tab \tab \tab \tab for all $j \in \mathbb{N}$ such that $1 \leq j \leq i:   E@j = E'@\pi(j)$ \hfill 1,3,4,lemma \ref{lemma:func}\\
6.&\tab \tab \tab $j,k \in \mathbb{N}$ and $E@j \in \textsc{Ret}$ and $E@k \in \textsc{Inv}$ and $j < k \leq i$ \hfill hyp.\\
7.&\tab \tab\ \tab \tab $e \in \textsc{Ret}$ and $e' \in \textsc{Inv}$ and pos(e,E) = j and pos(e',E) = k \hfill 6 \\
8.& \tab \tab \tab \tab$(e,e') \in obs(E,i)$ \hfill 6,7 \\
9.& \tab \tab \tab \tab$obs(E,i) \subseteq  obs(E', i')$ \hfill 3, def. $\sim_{\mathcal{A}}$, def. $\sim_{obs}$\\
10.& \tab \tab \tab \tab$(e,e') \in obs(E', i')$\\
11.& \tab \tab \tab \tab pos(e,E') $<$ pos(e',E') $\leq i'$\hfill 10\\
12.& \tab \tab \tab \tab for all $j \in \{ 1, \ldots ,i'\}$ there is $k \in \{ 1, \ldots, i\}$: $j = \pi(k)$ \hfill 5,$\pi$ surjective\\
13.& \tab \tab \tab \tab  $j',k' \in  \{ 1, \ldots, i\}$ and $E'@\pi(j')=e$ and $E'@\pi(k')=e'$ and $\pi(j') < \pi(k') \leq i'$ \hspace{1cm} \hfill 11,12\\
14.& \tab \tab  \tab \tab $E'@\pi(j)=e$ and $E'@\pi(k)=e'$ \hfill 5,7 \\
15.& \tab \tab \tab \tab $E'@\pi(j')= E'@\pi(j)$ and $E'@\pi(k')=E'@\pi(k)$ \hfill 13,14 \\
16.& \tab \tab \tab \tab $\pi(j')= \pi(j)$ and $\pi(k') = \pi(k)$ \hfill 2,15\\
17.& \tab \tab \tab \tab $\pi(j) < \pi(k) \leq i'$\\
18.& \tab \tab \tab for all $j,k \in \mathbb{N}$ if $E@j \in \textsc{Ret}$ and $E@k \in \textsc{Inv}$ and $j < k$ then $\pi(j) < \pi(k)$ \hfill 6,17\\
19.& \tab \tab $(E,i)\preceq_{real} (E',i')$ \hfill 5,18\\
20.& For all $(E,i) \in \mathcal{E}^\omega \times \mathbb{N}$: \\
& if unique then if there is $(E',i') \in \mathcal{E}^\omega \times \mathbb{N}:$  s.t. $(E,i)\sim_{\mathcal{A}}(E',i')$ and $(E',i') \models correct$  then lin (E,i).\hspace{0.5cm} \hfill 1-4,19
\end{tabular}
}
\end{center}
\end{table*}
\end{landscape}

%\subsection{TSO (Theorem \ref{thm:tso})}
%\label{app:TSO_proof}
%The operational and the logical description of TSO are equivalent: 
%\begin{center}
%$(E \downharpoonright i) \in \mathcal{T}_{TSO}$ iff $(E,i) \models
%\mathit{correctTSO}$
%\end{center}
%\input{proofTSO}
\subsection{Eventual Consistency (Theorem \ref{thm:evCons})}
% \label{app:Proof_EVC}
% \begin{center}
% $evCons(E)$ if and only if $E \models \neg D_{\textsc{Threads}} \neg(\mathit{correctEVC})$
% \end{center}

\label{proof:evCons}
The axiomatic and the logical description of eventual consistency are equivalent:
\begin{equation*}
\text{Show } \mathit{evCons}(E) \text{ iff } (E,\omega) \models \neg D_{\textsc{Threads}} \neg \specEVC
\end{equation*}
We restate the axiomatic definition for reference: 
\begin{definition}[Eventual Consistency]
\begin{enumerate}
\item $\prec_{v}  \subseteq \prec_{a}$ (arbitration extends visibility). 
\item $\prec_{p} \subseteq \prec_{v}$ (visibility is compatible with program-order).
\item for each $e_q = (t,\mathit{qu}(\mathit{id},q,r)) \in E$, we have r = $\mathit{apply}(\{e \; | \; e \prec_{v}  e_q \},\prec_{a},s_0)$ (consistent query results).
\item $\prec_{a}$ and $\prec_{v}$ factor over $\equiv_t $ ( atomic revisions).
\item if $(t, \mathit{com}(\mathit{id})) \not\in E$ and $(t,\_(\mathit{id},\_)) \prec_{v} (t',\_)$ then $t=t'$ (uncommitted updates).
%if $(t,\mathit{com},\mathit{id}) \notin E$, $\mathit{trans}(e)=(t,id)$, $\mathit{trans}(e')=(t',id')$ and $e \prec_{v} e'$ then $t=t'$.
\item if $e= (t,\mathit{com}(\mathit{id})) \in \mathit{E}$ then there are only finitely many $e' := (t',\mathit{com}(\mathit{id}'))$ such that $e' \in E$ and $e not\prec_{v} e'$ (eventual visibility).
\end{enumerate} 
\end{definition}
We will reference the parts of the definition by their numbers.
\begin{proof}
We write that $e$ is before $e'$ in $E$, if $\mathit{pos}(e,E) < \mathit{pos}(e',E) < \omega$. Let $\mathit{rev}(e)=(t,\mathit{id})$ :iff $e=(t,\_(id,\_)$. The revision number of an event $e$ is the revision $\mathit{rev}(e')$ of the next event $e':=(t,\mathit{com}(\mathit{id}))$ in $E$. We lift the relations $\prec_a$ and $\prec_v$ to revisions, i.e. $(t,id) \prec (t',id')$, if for the respective commit events $e=(t,\mathit{com}(\mathit{id}))$, and $e'=(t',\mathit{com}(\mathit{id'}))$ it holds that $e \prec e'$. Let $\mathit{align}(E,<)$ be a function that arranges a set of events $E$ according to order $<$.\\
\\
%We say that a query or update event is in revision $(t,\mathit{id})$ if $(t,\mathit(com)(\mathit{id}))$ is the next commit event in $E'$.
Show " $\rightarrow$": \\
Let $ E':=\mathit{align}(E,\prec_a)$, where we erase all $(\_,fwd(\_))$ events. Because $\prec_p \subseteq \prec_v \subseteq \prec_a$, we have $E \sim_\textsc{Threads} E'$. We now show that $(E',\omega) \models \specEVC$.

For each revision $(t,id)$ and for all $t' \neq t$ insert $(d,\mathit{fwd}(t,t',\mathit{id}))$ directly before the first event of revision $(t',\mathit{id'})$ if and only if $(t',id')$ is the minimal revision with respect to $\preceq_a$ such that $(t,\mathit{id}) \preceq_v (t',\mathit{id'})$ and such a revision exists.  If there are several forward events order them by $\preceq_a$ with respect to the revisions that were forwarded.\\
$(E,i) \models \mathit{atomicTrans}$, by (4), our assumption that committed revision ids match the revision ids of previous actions and our construction, as we add forward events only directly after commit events.\\
$(E',i) \models \mathit{fwd}$ follows by (1). Because visibility order must not contradict arbitration order, a revision becomes visible only after being committed. 
\\
Show $(E',\omega) \models \mathit{alive}$: Assume $(t,\mathit{com}(\mathit{id})) \in E'$ and there is $t'$ s.t. $t'$ commits infinitely often. Suppose that $(t,id)$ is never forwarded to $t'$, i.e. $(\mathit{fwd},t,t',id) \notin E'$. But then, by construction of $E'$, there are infinitely many  $(t',\_)$ such that $(t,id) \not\preceq_v (t',\_)$. This cannot be by condition (6).\\
\\
Show $(E',\omega) \models \specEVC$:\\
Assume $(E',k) \models$ query$(t,q,\mathit{res})$. Let $\mathcal{L}:= \mathit{act}(\mathit{align}(\{ e \; | \; e \preceq_v E'@k \}, \preceq_a))$. We need to show that: 
\begin{equation*}
(E',k) \models \mathcal{L} \; \mathit{validLog} \;  t \wedge \mathit{result}(q,\mathcal{L},r)
\end{equation*}
We show $(E',k) \models \mathcal{L} \; \mathit{validLog} \;  t$.
We get $(E',k) \models \mathit{consistent}(\mathcal{L})$ by the fact that we align with respect to $\prec_a$.
 Show $(E',k) \models \forall a ( \;t \; k_{log} \; a  \rightarrow a \in \mathcal{L} )$. Assume $(E',k) \models  \;t \; k_{log} \; a $, and $E'@k$ has revision id $(t,id)$. Then either there is $j$ such that $(E',j) =(t,a)$, and $a \in \mathcal{L}$, by (2), or there is $j$ such that $E@j=(d,\mathit{fwd}(t',t,\mathit{id'}))$ and there is $l < j$ such that $E@l= (t',\mathit{com}(\mathit{id'}))$, and $(E',l) \models t' k_{log} a$. By our construction, we have $(t',\mathit{id'}) \prec_v (t,\mathit{id})$. Then, by the same argument and the transitivity of $\prec_v$, we have $a \in \mathcal{L}$.\\
\\
Show $(E',k) \models \forall a (a \in \mathcal{L} \rightarrow \;t \; k_{log} \; a)$.
Assume $e \preceq_a a$ and $e \preceq_v a$. Let $a$ be in revision $(t,\mathit{id})$, and $e$ in revision $(t',id')$. Assume $t=t'$. Then $(E',k) \models t \; k_{log} \; a$ by (2) and the definition of $k_{log} $. Assume $t \neq t'$. 
Then either $(t,\mathit{id})$ is the the earliest revision of $t$ such that $(t',\mathit{id}') \preceq_v (t,id)$, and we have entered $(d,\mathit{fwd}(t',t,id'))$ before the first event of $(t,\mathit{id})$, or there is an earlier revision in which case we inserted  $(d,\mathit{fwd}(t',t,id'))$ before. By the definition of $k_{log}$, we have $(E',k) \models t \; k_{log} a$.\\
\\
Having established that $\mathcal{L}$ is valid, $(E',k) \models \mathit{result}(q,\mathcal{L},\mathit{r})$ follows by definition.\\
\\
Show "$\leftarrow$":\\
We have $E \sim_{\textsc{Threads}} E'$ and $E' \models \specEVC$. Let $e \preceq_a e'$ iff $\mathit{pos}(e,E') < \mathit{pos}(e',E')$ and $e:=(\_,a) \preceq_v e':=(t,\_)$ iff there is $\mathcal{L}$ such that $(E',\mathit{pos}(e')) \models \mathcal{L} \; \mathit{ validLog } \; t \wedge \mathit{result}(q,\mathcal{L},\mathit{res})$, and $a \in \mathcal{L}$. \\
Show $\preceq_a$ is a total order: This follows from the definition of $\mathit{pos}$ and the fact that revision ids occur only once. Show $\preceq_v$ is a partial order: (a) Show $e \preceq_v e$: This follows from the definition of $k_{log}$, as threads know their own actions. (b) Show antisymmetry: follows from the definition of $\mathit{pos}$ and the fact that transactions occur only once. (c) Show transitivity: follows from the transitivity of $\prec_p$, and the recursive definition of $k_{log}$. We now prove the individual parts of the definition. 
$(1)$: Follows by $(E',\mathit{pos}(e')) \models \mathit{consistent}(\mathcal{L})$.
$(2)$: By our definition of $k_{log}$ i.e. the fact that threads know their own actions.
$(3):$ Follows by $(E',\mathit{pos}(e')) \models \mathit{result}(q,\mathcal{L},\mathit{res})$.
$(4):$ Follows from $E' \models \mathit{atomicTrans}$.
$(5):$ by $E' \models \mathit{fwd}$, as updates are only forwarded after being committed.
$(6):$ Follows from $E' \models \mathit{alive}$. 

% Suppose there are infinitely many revisions s.t. $(t,id) \not\preceq_v (t',id')$. Then there must be infinitely many revisions some thread $t'$, s.t. $(t,id) \preceq_v (t',\_)$. But then by evntReceive, we get $(t',rec,t,id)$ at some future point. But from this point on, $(t,id)$ is visible, leading to a contradiction.
\end{proof}

\subsection{Knowledge about Consistency (Theorems \ref{thm:dec_seqs}
  and \ref{thm:dec_lin})}
We restate the relevant axioms for reference:
Everything a group of thread knows is also true:
$\text{(T)} := \; \models D_G \varphi \rightarrow \varphi \; \text{(Truth axiom)}$,
groups of threads know what they know:
$\text{(4)} := \; \models D_G \varphi \rightarrow D_G D_G \varphi \; \text{ (positive introspection)}$
and groups of threads know what they do not know:
$\text{(5)} := \; \models \neg D_G \varphi \rightarrow D_G \neg D_G \varphi \;  \text{ (negative introspection)}$.

Threads can decide whether a trace is sequentially consistent or not:
\begin{equation*}
  \begin{array}[t]{@{}l@{}}
    \models (\seqCons \leftrightarrow D_\mathit{Threads}(\seqCons)) \mathrel{\land}\\[\jot]
    (\neg\seqCons \leftrightarrow D_\mathit{Threads}(\neg\seqCons))
  \end{array}
\end{equation*} 
\begin{proof}
By instantiating axiom (5) with $\varphi:= \neg \specCorrect $ and $G:=\textsc{Threads}$, we get:
\[1. \models \neg D_\textsc{Threads} (\neg \specCorrect)
\mathrel{\rightarrow} 
D_\textsc{Threads} (\neg D_\textsc{Threads} (\neg \specCorrect)) \]
We get $2.$ from $1.$ by applying the definition of $\seqCons$, (i.e. $\seqCons:= \neg \DThreads (\neg correct$)):
\[ 2. \models \seqCons  \rightarrow \DThreads (\seqCons) \]
We get $3$ by instantiating (T)  with $\varphi :=
\seqCons$ and $G:=\textsc{Threads}$:
\[ 3.  \models \DThreads (\seqCons) \rightarrow \seqCons. \]
This proves the first conjunct. The second conjunct is proved in a
similar way, by instantiating (4).
\end{proof}

There are linearizable traces on which the threads together with the observer do not know that they are linearizable. There is $E \in \mathcal{E}^\infty$ such that 
\begin{equation*}
  \begin{array}[t]{@{}l@{}}
  (E,i) \models \mathit{Lin} \land \neg D_\mathit{Threads \uplus \{ \mathit{obs}\}}(\mathit{Lin})
  \end{array}
\end{equation*} 
As in sequential consistency, the threads together with the observer can spot when a trace is not linearizable.
\begin{equation*}
\models  \neg \mathit{Lin} \leftrightarrow D_\mathit{Threads  \uplus \{ \mathit{obs}\}}(\neg\mathit{Lin})
\end{equation*}
\begin{proof}
As $\preceq_\mathit{obs}$ is a partial order, axiom (5) for negative introspection is not a validity. That means its negation is satisfiable. The proof for the second claim is analogue to the case of sequential consistency.
\end{proof}

% \Ifthenelse{\Equal{\isTechReport}{false}}{
%   \input{proofTSO}
% }{
%   \input{proofLin}
%   \input{proofTSO}
%   \input{proofEvCons}
% }

\end{document}